\documentclass{lmcs} 
\usepackage[utf8]{inputenc}
\pdfoutput=1

\usepackage{lastpage}
\lmcsdoi{18}{3}{13}
\lmcsheading{}{\pageref{LastPage}}{}{}%
{Dec.~20,~2018}{Jul.~29,~2022}{}


\usepackage{algorithm}
\usepackage{algorithmic}
\usepackage{amsmath}
\usepackage{amsthm}
\usepackage{amsmath}
\usepackage{amsfonts}
\usepackage{amssymb}
\usepackage{comment}
\usepackage{enumitem}
\usepackage{graphicx}
\usepackage{hyperref}
\usepackage{listings}

\usepackage[graphicx]{realboxes}
\usepackage{tikz}
\usetikzlibrary{automata,positioning}
\usepackage{tikz-qtree}
\usepackage{float}
\usepackage{floatflt}

\renewcommand*{\cc}[1]{\mathsf{#1}}

\newcommand*{\N}{\mathbb{N}}
\newcommand*{\Q}{\mathbb{Q}}
\renewcommand*{\L}{\mathcal{L}}
\newcommand*{\I}{\mathcal{I}}

\newcommand*{\A}{\mathcal{A}}
\newcommand*{\G}{\mathcal{G}}
\newcommand*{\Obs}{\mathit{O}}

\newcommand*{\Statess}{\mathit{S}}
\newcommand*{\Final}{\mathcal{F}}
\newcommand*{\StartState}{{\mathit{Init}}}
\newcommand*{\AcceptingStates}{\Final}

\newcommand*{\LS}[1]{\mathsf{LimSup(#1)}}

\newcommand*{\LSAut}{\A_{\succ \mathit{LS}}}

\newcommand*{\DSum}[2]{\mathit{DS}({#1}, {#2})}

\newcommand*{\DSsucceqFord}[1]{#1}
\newcommand*{\DSsucceqFordf}[1]{\succ_{DS(d)}}

\newcommand*{\MaxC}{\mathit{maxC}}
\newcommand*{\MaxX}{\mathit{maxX}}

\newcommand*{\DSumDiff}[1]{\mathit{DS}^{-}(B, A, d, #1)}

\newcommand*{\Res}{\mathit{Res}}
\newcommand*{\floor}[1]{\lfloor #1 \rfloor}
\newcommand*{\ceil}[1]{\lceil #1 \rceil}
\newcommand*{\CSum}{\mathit{CSum}}

\renewcommand*{\lim}[3]{\text{lim}_{#1 \rightarrow #2} #3}
\renewcommand*{\liminf}[3]{\text{lim inf}_{#1 \rightarrow #2} #3}
\renewcommand*{\limsup}[3]{\text{lim sup}_{#1 \rightarrow #2} #3}
\newcommand*{\LA}[1]{\mathsf{LimAvg}(#1)}
\newcommand*{\PLA}[1]{\mathsf{PrefixAvg}(#1)}
\newcommand*{\LASup}[1]{\mathsf{LimSupAvg}(#1)}
\newcommand*{\LAInf}[1]{\mathsf{LimInfAvg}(#1)}
\newcommand*{\Av}[1]{\mathsf{Avg}(#1)}
\newcommand*{\PASuc}{\mathsf{PA}}

\newcommand*{\Sum}[1]{\mathsf{Sum}(#1)}
\newcommand*{\ExistFin}[3]{\exists^{f} #1,  \Sum{#2[0,#1-1]} \geq \Sum{#3[0,#1-1]}} 
\newcommand*{\ExistInf}[3]{\exists^{\infty} #1,  \Sum{#2[0,#1-1]} > \Sum{#3[0,#1-1]}} 

\newcommand*{\less}{\mathit{DomProof}}
\newcommand*{\first}{\mathit{Dom}}

\newcommand*{\DSInclusion}{\mathsf{InclusionReg}}
\newcommand*{\AugmentWtAndLabel}{\mathsf{AugmentWtAndLabel}}

\newcommand*{\MakeProduct}{\mathsf{MakeProduct}}
\newcommand*{\Intersect}{\mathsf{Intersect}}
\newcommand*{\Project}{\mathsf{FirstProject}}

\newcommand*{\SepRuns}{\mathsf{UniqueId}}
\newcommand*{\CompareRuns}{\mathsf{Compare}}
\newcommand*{\EnsureRuns}{\mathsf{DimEnsure}}

\newcommand*{\cct}[1]{\textsf{#1}}
\renewcommand{\Re}{\mathbb{R}}

\begin{document}
\title[Comparator automata]{Comparator automata in quantitative verification}
\titlecomment{An earlier version of this paper has appeared at FoSSaCS 2018~\cite{BCVFoSSaCS18}}

\author[S.~Bansal]{Suguman Bansal\lmcsorcid{https://orcid.org/0000-0002-0405-073X}}[a]	
\address{University of Pennsylvania, Philadelphia, USA}	
\email{suguman@seas.upenn.edu}  

\author[S.~Chaudhuri]{Swarat Chaudhuri\lmcsorcid{https://orcid.org/0000-0002-6859-1391}}[b]	
\address{University of Texas, Austin, USA}	
\email{swarat@cs.utexas.edu}  

\author[M.Y.~Vardi]{Moshe Y. Vardi\lmcsorcid{https://orcid.org/0000-0002-0661-5773}}[c]	
\address{Rice University, Houston, USA}	
\email{vardi@cs.rice.edu}  

\begin{abstract}
  The notion of comparison between system runs is fundamental in formal verification.  This concept is implicitly present in the verification of qualitative systems, and is more pronounced in the verification of quantitative systems.  In this work, we identify a novel mode of comparison in quantitative systems: the online comparison of the aggregate values of two sequences of quantitative weights. This notion is embodied by \emph{comparator automata} (\emph{comparators}, in short), a new class of automata that read two infinite sequences of weights synchronously and relate their aggregate values.

We show that {aggregate functions} that can be represented with B\"uchi automaton result in
comparators that are finite-state and accept by the B\"uchi condition as well. Such \emph{$\omega$-regular comparators} further lead to generic algorithms for a number of well-studied problems, including the quantitative inclusion and winning strategies in quantitative graph games with incomplete information, as well as related non-decision problems, such as obtaining a finite representation of all counterexamples in the quantitative inclusion problem.

We study comparators for two aggregate functions: discounted-sum and limit-average.  We prove that the discounted-sum comparator is $\omega$-regular iff the discount-factor is an integer.
Not every aggregate function, however, has an $\omega$-regular comparator. Specifically, we show that the language of sequence-pairs for which limit-average aggregates exist is neither $\omega$-regular nor $\omega$-context-free.  Given this result, we introduce the notion of \emph{prefix-average} as a relaxation of limit-average aggregation, and show that it admits $\omega$-context-free comparators i.e.\ comparator automata expressed by B\"uchi pushdown automata.

\end{abstract}

\maketitle

\section{Introduction}%
\label{Sec:Intro}
Many classic questions in formal methods can be seen as involving \emph{comparisons} between different system runs or inputs.  Consider the problem of verifying if a system $S$ satisfies a linear-time temporal property $P$. Traditionally, this problem is phrased language-theoretically: $S$ and $P$ are interpreted as sets of (infinite) words, and $S$ is determined to satisfy $P$ if $S \subseteq P$. The problem, however, can also be framed in terms of a \emph{comparison} between words in $S$ and $P$.
Suppose a word $w$ is assigned a weight of 1 if it belongs to the language of the system or property, and 0 otherwise. Then determining if $S \subseteq P$ amounts to checking whether the weight of every word in $S$ is less than or equal to its weight in $P$~\cite{baier2008principles}.

The need for such a formulation is clearer in quantitative systems, in which every run of a word is associated with a sequence of (rational-valued) weights. The weight of a run is given by \emph{aggregate function} $f: \mathbb{Q}^{\omega}\rightarrow \mathbb{R}$, which
returns the real-valued \emph{aggregate value} of the run's weight sequence.
The weight of a word is given by the supremum or infimum of the weight of all its runs.
Common examples of aggregate functions include discounted-sum and limit-average.

In a well-studied class of problems involving quantitative systems, the objective is to check if the aggregate value of words of a system exceed a constant threshold value~\cite{de2004model,de2004linear,degorre2010energy}.
This is a natural generalization of emptiness problems in  qualitative systems.
Known solutions to the problem involve arithmetic reasoning via linear programming and graph algorithms such as negative-weight cycle detection,  computation of maximum weight of cycles etc~\cite{andersson2006improved,karp1978characterization}.

A more general notion of comparison relates  aggregate values of two weight sequences.  Such a notion arises in the \emph{quantitative inclusion problem} for weighted automata~\cite{AlmagorBokerKupferman},
where the goal is to determine whether the weight of words in one weighted automaton is less than that in another.
Here it is necessary to compare the aggregate value along runs between the two automata.
Approaches based on  arithmetic reasoning do not, however, generalize to solving such problems.
In fact, the known solution to discounted-sum inclusion with integer discount-factor combines linear programming with a \emph{specialized} subset-construction-based determinization step, rendering an \cct{EXPTIME} algorithm~\cite{andersson2006improved,boker2014exact}.
Yet, this approach does not match the \cct{PSPACE} lower bound for discounted-sum inclusion.

In this paper, we present an automata-theoretic formulation of this form of comparison between weighted sequences. Specifically, we introduce \emph{comparator automata} (\emph{comparators}, in short), a class of automata that read pairs of infinite weight sequences synchronously, and compare their aggregate values in an online manner. While comparisons between weight sequences happen implicitly in prior approaches to quantitative systems, comparator automata make these comparisons explicit. We show that this has many benefits, including generic algorithms for a large class of quantitative reasoning problems, as well as a direct solution to the problem of discounted-sum inclusion that also closes its complexity gap.

A \emph{comparator for aggregate function $f$ for relation $R$} is an automaton that accepts a pair $(A,B)$ of sequences  of bounded natural numbers iff $f(A)$ $R$ $f(B)$, where $R$ is an inequality relation ($>$, $<$, $\geq$, $\leq$, $\neq$) or the equality relation $=$.
A comparator could be finite-state or  (pushdown) infinite-state.
This paper studies such comparators.

A comparator is \emph{$\omega$-regular} if it is finite-state and accepts by the B\"uchi condition. We relate $\omega$-regular comparators to \emph{$\omega$-regular aggregate functions}~\cite{chaudhuri2013regular}, and show that $\omega$-regular aggregate-functions entail $\omega$-regular comparators.
However, the other direction is still open: Does an $\omega$-regular comparator for an aggregate function and a relation imply that the aggregate function is also $\omega$-regular? Furthermore,
we show that $\omega$-regular comparators lead to generic algorithms for a number of well-studied problems including the quantitative inclusion problem, and in solving quantitative games with incomplete information.
Our algorithm yields \cct{PSPACE}-completeness of quantitative inclusion when the $\omega$-regular comparator is provided.
The same algorithm extends to obtaining finite-state representations of counterexample words in inclusion.

Next, we show that the discounted-sum aggregation function admits an $\omega$-regular comparator for all relations $R $ iff the discount-factor $d>1$ is an integer. We use this result to prove that discounted-sum aggregate function for discount-factor $d>1$ is $\omega$-regular iff $d$ is an integer.
Furthermore, we use properties of $\omega$-regular comparators to conclude that the discounted-sum inclusion is \cct{PSPACE}-complete, hence resolving the complexity gap (under a unary representation of numbers).

Finally, we investigate the limit-average comparator. Since limit-average is only defined for sequences in which the average of prefixes converge, limit-average comparison is not well-defined.
We show that even a B\"uchi pushdown automaton cannot separate sequences for which
 limit-average exists  from those for which it does not.
Hence, we introduce the novel notion of \emph{prefix-average comparison} as a relaxation of limit-average comparison. We show that the prefix-average comparator admits a comparator that is  $\omega$-context-free, i.e., given by a B\"uchi pushdown automaton,
and we discuss the utility of this characterization.

This paper is organized as follows: Preliminaries are given in \textsection~\ref{Sec:Prelims}.
Comparator automata are formally defined in \textsection~\ref{Sec:Comparator}.
The connections between $\omega$-regular aggregate functions and $\omega$-regular comparators is discussed in Section~\ref{Sec:RegularAggFunctions}.
Generic algorithms for $\omega$-regular comparators are discussed in \textsection~\ref{Sec:Inclusion}-\ref{Sec:GraphGames}.
\textsection~\ref{Sec:DiscountedSum} discusses discounted-sum aggregate function and its comparators with non-integer rational discount-factors (\textsection~\ref{Sec:RationalDF}) and integer discount-factors (\textsection~\ref{Sec:IntegerDF}).
The construction and properties of prefix-average comparator are given in \textsection~\ref{Sec:LA}.
We conclude with future directions in \textsection~\ref{Sec:Conclusion}.

\subsection{Related work}%
\label{Sec:RelatedWork}
The notion of comparison has been widely studied in quantitative settings. Here we mention only a few of them. Such aggregate-function based notions appear in weighted automata~\cite{AlmagorBokerKupferman,droste2009handbook}, quantitative games including mean-payoff and energy games~\cite{degorre2010energy}, discounted-payoff games~\cite{andersen2013fast,andersson2006improved}, in systems regulating cost, memory consumption, power consumption, verification of quantitative temporal properties~\cite{de2004model,de2004linear},  and others. Common solution approaches include graph algorithms such as weight of cycles or presence of cycle~\cite{karp1978characterization}, linear-programming-based approaches,  fixed-point-based approaches~\cite{chatterjee2010energy}, and the like.
The choice of approach for a problem typically depends on the underlying aggregate function. In contrast, in this work we present an automata-theoretic approach that unifies solution approaches to problems on  different aggregate functions. We identify a class of aggregate functions, ones that have an $\omega$-regular comparator, and present generic algorithms for some of these problems.

While work on finite-representations of counterexamples and witnesses in the qualitative setting is known~\cite{baier2008principles}, we are not aware of such work in the quantitative verification domain.
This work can be interpreted as automata-theoretic  arithmetic, which has been explored in regular real analysis~\cite{chaudhuri2013regular}.

\section{Preliminaries}%
\label{Sec:Prelims}
Let $\Sigma$ be a finite set of alphabet. The set of finite and infinite words over $\Sigma$ is denoted by $\Sigma^*$ and $\Sigma^\omega$, respectively. An \emph{aggregate function} $f: \mathbb{Q}^\omega \rightarrow \Re$ takes the aggregate of an infinite-length weight sequence.

\begin{defi}[B\"uchi automaton~\cite{thomas2002automata} ]

A (finite-state) \emph{B\"uchi automaton} is a tuple
  $\A = (\Statess$, $\Sigma$, $\delta$, $\StartState$, $\Final)$, where
  $ \Statess $ is a finite set of \emph{states}, $ \Sigma $ is a finite \emph{input alphabet},   $ \delta \subseteq (\Statess \times \Sigma \times \Statess) $ is the   \emph{transition relation}, $ \StartState \subseteq \Statess $ is the set of \emph{initial states}, and $ \Final \subseteq \Statess $ is the set of \emph{accepting states}.
\end{defi}
A B\"uchi automaton is \emph{deterministic} if for all states $ s $ and
inputs $a$, $ |\{s'|(s, a, s') \in \delta \textrm{ for some $s'$} \}|
\leq 1 $ and $|\StartState|=1$. Otherwise, it is \emph{nondeterministic}.
A B\"uchi automaton is \emph{complete} if for all states $ s $ and
inputs $a$, $ |\{s'|(s, a, s') \in \delta \textrm{ for some $s'$} \}| \geq 1$.
For a word $ w = w_0w_1\dots \in \Sigma^{\omega} $, a \emph{run} $ \rho$ of $ w $ is a sequence of states $s_0s_1\dots$ s.t.
$ s_0 \in \StartState$, and $ \tau_i =(s_i, w_i, s_{i+1}) \in \delta $ for all $i$.
Let $ \mathit{inf}(\rho) $ denote the set of states that occur infinitely
often in run ${\rho}$.
A run $\rho$ is an \emph{accepting run} if $ \mathit{inf}(\rho)\cap \Final \neq \emptyset $. A word $w$ is an \emph{accepting word} if it has an accepting run.
B\"uchi automata are  closed under set-theoretic union,
intersection, and complementation~\cite{thomas2002automata}. Languages accepted by these automata are called \emph{$ \omega $-regular   languages}.

 \subsection*{Reals over \texorpdfstring{$\omega$}{omega}-words\texorpdfstring{~\cite{chaudhuri2013regular}}{}}

  Given an integer base $\beta\geq 2$, its \emph{digit set} is $\mathsf{Digit}(\beta) =\{0,\dots, \beta-1\}$. Let $x \in \Re$, then there exist unique words $\mathsf{Int}(x, \beta) = z_0z_1\dots \in \mathsf{Digit}(\beta)^*\cdot 0^\omega $ and
  $\mathsf{Frac}(x, \beta)=f_0f_1\dots \in \mathsf{Digit}(\beta)^\omega \setminus \mathsf{Digit}(\beta)^*\cdot (\beta-1)^\omega$ such that $|x| = \Sigma_{i=0}^\infty \beta^i \cdot z_i + \Sigma_{i=0}^\infty \frac{f_i}{\beta^i}$.
 Thus, $z_i$ and $f_i$ are respectively the $i$-th
 least significant digit in the base $\beta$ representation of the
 integer part of $x$, and the $i$-th most significant digit in
 the base $\beta$ representation of the fractional part of $x$.
 Then, a real-number $x \in \Re$ in base $\beta$ is represented by $\mathsf{rep}(x, \beta) = \mathsf{sign} \cdot  (\mathsf{Int}(x, \beta),\mathsf{Frac}(x, \beta) )$, where $\mathsf{sign} = +$ if $x \geq 0$, $\mathsf{sign} = -$ if $x<0$, and $(\mathsf{Int}(x, \beta),\mathsf{Frac}(x, \beta) )$ is the interleaved word of $\mathsf{Int}(x, \beta)$ and $\mathsf{Frac}(x, \beta)$.
 Clearly, $x = \mathsf{sign}\cdot |x| = \mathsf{sign}\cdot (\Sigma_{i=0}^\infty \beta^i \cdot z_i + \Sigma_{i=0}^\infty \frac{f_i}{\beta^i})$.
 For all integer $\beta\geq 2$, we denote the alphabet of representation of real-numbers in base $\beta$ by $\mathsf{AlphaRep}(\beta) = \{+,-\}\cup \mathsf{Digit}(\beta)\times\mathsf{Digit}(\beta)$.

 We adopt the definitions of function automata and regular functions~\cite{chaudhuri2013regular} w.r.t.\ aggregate functions as follows:
 \begin{defi}[Aggregate function automaton, $\omega$-Regular aggregate function]%
 	\label{def:aggregateregular}
 	Let $\Sigma$ be a finite set, and $\beta\geq 2$ be an integer-valued base.
 	A B\"uchi automaton $\A$ over alphabet $\Sigma\times\mathsf{AlphaRep}(\beta)$ is an \emph{aggregate function automaton of type $\Sigma^\omega\rightarrow\Re$} if for all $A \in \Sigma^\omega$, there exists exactly one $x \in \Re$ such that $(A,\mathsf{rep}(x,\beta)) \in \L(\A)$.

 \end{defi}
 Equivalently, the language could be represented by a Parity automaton instead of a B\"uchi automaton.

 $\Sigma$ and $\mathsf{AlphaRep}(\beta)$ are the \emph{input} and \emph{output} alphabets, respectively. An aggregate function is an arbitrary function $f:\Sigma^\omega\rightarrow \Re$. An aggregate function $f:\Sigma^\omega\rightarrow \Re$ is said to be \emph{$\omega$-regular} under integer base $\beta \geq 2$ if there exists an aggregate function automaton $\A$ over alphabet $\Sigma\times\mathsf{AlphaRep}(\beta)$ such that for all sequences $A \in \Sigma^\omega$ and  $x\in \Re$, $f(A) = x$ iff $(A,\mathsf{rep}(x,\beta)) \in L(\A)$.

\begin{defi}[Weighted automaton~\cite{chatterjee2010quantitative,mohri2009weighted}]
A \emph{weighted automaton} over infinite words is a tuple $\A= (\mathcal{M},  \gamma, f)$, where $\mathcal{M} = (\Statess, \Sigma, \delta, \StartState, \Statess) $ is a B\"uchi automaton with all states as accepting, $\gamma: \delta \rightarrow \Q$ is a \emph{weight function}, and $f: \Q \rightarrow \Re$ is the \emph{aggregate function}.
\end{defi}
\emph{Words} and \emph{runs} in weighted automata are defined as they are in B\"uchi automata.
The \emph{weight-sequence} of run $\rho = s_0 s_1 \dots$ of word $w = w_0w_1\dots$ is given by $wt_{\rho} = n_0 n_1 n_2\dots$ where $n_i = \gamma(s_i, w_i, s_{i+1})$ for all $i$.
The \emph{weight of a run}  $\rho$, denoted by $f(\rho)$, is given by $f(wt_{\rho})$.
Here the \emph{weight of a word} $w \in \Sigma^{\omega}$ in weighted automata is defined as $wt_{\A}(w) = sup \{f(\rho) | \rho$ is a run of $w$ in $\A \}$.
In general, weight of a word can also be defined as the infimum of the weight of all its runs. Note, an automaton need not accept every word, even though all its states are accepting, since it need not be complete.
By convention, if a word $w \notin \L(\mathcal{M})$ its weight $wt_\A(w) = -\infty$.


\begin{defi}[Quantitative inclusion]
Let $P$ and $Q$ be  weighted $\omega$-automata with the \emph{same} aggregate function $f$.
The \emph{strict quantitative inclusion problem}, denoted by  $P \subset_f Q$,  asks whether  for all words $w \in \Sigma^{\omega}$, $wt_{P}(w) < wt_{Q}(w)$.
The \emph{non-strict quantitative inclusion problem},  denoted by  $P \subseteq_f Q$,  asks whether  for all words $w \in \Sigma^{\omega}$, $wt_{P}(w) \leq wt_{Q}(w)$.

\end{defi}
 Quantitative inclusion, strict and non-strict, is \cct{PSPACE}-complete for limsup and liminf~\cite{chatterjee2010quantitative}. Non-strict quantitative inclusion is undecidable for limit-average~\cite{degorre2010energy}, while decidability of the strict variant is still open.
 For discounted-sum with integer discount-factor it is in \cct{EXPTIME}~\cite{boker2014exact,chatterjee2010quantitative}, and decidability is unknown for rational discount-factors.

\begin{defi}[Incomplete-information quantitative games]
An \emph{incomplete-information quantitative game} is a tuple $\G = (S,s_{\mathcal{I}} ,\Obs, \Sigma, \delta, \gamma, f)$, where $S$, $\Obs$, $\Sigma$ are sets of \emph{states}, \emph{observations}, and \emph{actions}, respectively, $s_{\mathcal{I}} \in S$ is the \emph{initial state}, $\delta \subseteq S \times \Sigma \times S$ is the \emph{transition relation},  $\gamma : S \rightarrow \N\times\N$ is the \emph{weight function}, and $f: \N^{\omega}\rightarrow \mathbb{R} $ is the \emph{aggregate function}.
\end{defi}
The transition relation $\delta$ is \emph{complete}, i.e., for all states $p$ and actions $a $, there exists a state $q $ s.t. $(p, a, q)\in \delta$.
A \emph{play} $\rho$ is a sequence $s_0a_0s_1a_1\dots$, where $\tau_i=(s_i, a_i, s_{i+1}) \in \delta$.
The \emph{observation of state} $s $, by abuse of notation, is denoted by $\Obs(s) \in \Obs$.
The \emph{observed play} $o_{\rho}$ of $\rho$ is the sequence
$o_0a_0o_1aa_1\dots$, where $o_i = \Obs(s_i)$.
Player $P_0$ has incomplete information about the game $\G$; it only perceives the observation play $o_{\rho}$. Player $P_1$ receives full information and witnesses play $\rho$.
Plays begin in the initial state $s_0 = s_\mathcal{I}$. For $i\geq 0$, Player $P_0$ selects action $a_i$. Next, player $P_1$ selects the state $s_{i+1}$, such that $(s_i, a_i, s_{i+1}) \in \delta$.
The \emph{weight of state} $s$ is the pair of payoffs $\gamma(s) =  (\gamma(s)_0, \gamma(s)_1)$.
The \emph{weight sequence} $wt_i$ of player $P_i$ along $\rho$ is given by $\gamma(s_0)_i \gamma(s_1)_i  \dots$,
and its payoff from $\rho$ is given by $f(wt_i)$ for aggregate function $f$, denoted by $f(\rho_i)$, for simplicity.
A play on which a player receives a greater payoff than the other player is said to be a \emph{winning play} for the player.
 A strategy for player $P_0$ is given by a function $\alpha : \Obs^* \rightarrow \Sigma$ since it only sees observations.  Player $P_0$ agrees with strategy $\alpha$ if for all $i$, $a_i = \alpha(o_0\dots o_i)$. A strategy $\alpha$ is said to be a  \emph{winning strategy} for player $P_0$ if all plays agreeing with $\alpha$ are winning plays for $P_0$.

  \begin{defi}[B\"uchi pushdown automaton~\cite{cohen1977theory}]
     A \emph{B\"uchi pushdown automaton (B\"uchi PDA)} is a tuple
    $\A = (\Statess, \Sigma, \Gamma, \delta, \StartState, Z_0, \Final)$, where $\Statess$, $\Sigma$, $\Gamma$, and $\Final$ are finite sets of \emph{states}, \emph{input alphabet},  \emph{pushdown alphabet} and \emph{accepting states}, respectively.
    $ \delta \subseteq (\Statess \times \Gamma\times (\Sigma\cup \{\epsilon\}) \times \Statess \times \Gamma) $ is the
    \emph{transition relation}, $ \StartState \subseteq \Statess $ is a
    set of \emph{initial states}, $Z_0 \in \Gamma$ is the \emph{start symbol}.
  \end{defi}
 A \emph{run} $\rho$ on a word $ w = w_0 w_1\dots\in\Sigma^{\omega} $ of a B\"uchi PDA $\A$ is a sequence of configurations $(s_0,\gamma_0),(s_1,\gamma_1)\dots$ satisfying (1) $ s_0 \in \StartState$, $\gamma_0 = Z_0$, and (2) ($s_i, \gamma_i, w_i, s_{i+1}, \gamma_{i+1}) \in\delta$ for   all $i$.
 B\"uchi PDA consists of a \emph{stack}, elements of which are the tokens $\Gamma$, and initial element $Z_0$.
Transitions \emph{push}  or \emph{pop} token(s) to/from the top of the stack.
  Let $ \mathit{inf}(\rho) $ be the set of states that occur infinitely
  often in state sequence  $s_0s_1\dots$ of run $\rho$.
 A run $\rho$ is  an \emph{accepting run} in B\"uchi PDA if $ \mathit{inf}(\rho)\cap \Final \neq \emptyset $.
  A word $w$ is an \emph{accepting word} if it has an accepting run.
Languages accepted by B\"uchi PDA are called \emph{$ \omega $-context-free languages} ({$\omega$-CFL}).

\subsection*{Notation and Terminology}
For an infinite sequence $A = (a_0, a_1, \dots)$, $A[i]$ denotes its $i$-th element, and $A[n,m]$ denotes the finite word $A[n]A[n+1]\dots A[m]$.
An infinite weight-sequence $A$ is said to be \emph{bounded} if there exists a value $b \in \mathbb{Q}$ such that $|A[i]| < b$ for all $i \geq 0$.
Abusing notation, we write $ w \in \A $ and $\rho \in \A$ if $w$ and $\rho$ are an accepting word and an accepting run of $\A$ respectively. The
Symbol $\cdot$ is used to denote both multiplication of real numbers and concatenation of sequences. The meaning will be clear in context.


\section{Comparator automata}%
\label{Sec:Comparator}
\emph{Comparator automata} (often abbreviated as \emph{comparators}) are a class of automata that can read pairs of weight sequences synchronously and  establish an equality or inequality relationship between these sequences.
Formally, we define:
\begin{defi}[Comparator automata]
	Let $\Sigma$ be a {finite set} of rational numbers, and $f:\Q^{\omega}\rightarrow \Re$ denote an aggregate function.
	A \emph{comparator automaton for aggregate function $f$} with inequality or equality relation $R \in \{\leq, <, \geq, >, =, \neq\}$ is an automaton over the alphabet $\Sigma\times\Sigma$  that accepts a pair $(A,B)$ of (infinite) weight sequences  iff $f(A)$ $R$ $f(B)$.
\end{defi}

From now on, unless mentioned otherwise, we assume that all weight sequences are bounded, natural number sequences.
The boundedness assumption is justified since  the set of weights forming the alphabet of a comparator is bounded.
For all aggregate functions considered in this paper, the result of comparison of weight sequences is preserved by a uniform linear transformation that converts rational-valued weights into natural numbers; justifying the natural number assumption.

When the comparator for an aggregate function and a relation is a B\"uchi automaton, we call it an \emph{$\omega$-regular comparator}. Likewise, when the comparator is a B\"uchi pushdown automaton, we call it an \emph{$\omega$-context-free comparator}.

\begin{thm}%
	\label{lemm:allregular}
	Let $\Sigma$ be a finite alphabet of natural numbers. Let $f : \N^{\omega}\rightarrow \Re$ be an aggregate function. If the comparator automata for aggregate function $f$ for any one inequality relation $R\in\{\leq, < , \geq, >\}$ is $\omega$-regular, then the comparator for $f$ is $\omega$-regular for all relations in $ \{\leq, <,\geq, >, = , \neq\}$.
\end{thm}
\begin{proof}[Proof sketch]
	Suppose the comparator automata $\A$ for relation $\leq$ is $\omega$-regular for aggregate function $f$. Then the pair of weight sequences $(A, B)\in \A$ iff $f(A) \leq f(B)$ holds. Since $f(A) \leq f(B)$ iff $f(B) \geq f(A)$ the automaton obtained by altering transition $s\xrightarrow{(a,b)} t$ in $\A$ to transition $s \xrightarrow{(b,a)} t$ is the $\omega$-regular comparator automata  for relation $\geq$.
	The $\omega$-regular comparator for $=$ can be obtained by taking the intersection of the comparator for $\leq$ and $\geq$ since B\"uchi automata are closed under intersection.
	Finally, since B\"uchi automata are also closed under complementation, we get that the comparator automata for the other three relations, namely $<$,$>$,$\neq$, is also $\omega$-regular.
\end{proof}
Later, we see that discounted-sum comparator is $\omega$-regular (\textsection~\ref{Sec:DiscountedSum}) and prefix-average comparator with $\geq$ (or $\leq$) is $\omega$-context-free (\textsection~\ref{Sec:LA}).

\subsection*{Limsup comparator}

\renewcommand{\max}{\mu}
We explain comparators through an example.  The \emph{limit supremum} (limsup, in short) of a bounded, integer sequence $A$, denoted by $ \LS{A}$, is the largest integer that appears infinitely often in $A$.
The \emph{limsup comparator for relation $\geq$} is a B\"uchi automaton that accepts the pair $(A,B)$ of sequences  iff  $\LS{A}\geq \LS{B} $.

The working of the limsup comparator for relation $\geq$ is based on non-deterministically guessing the limsup of sequences $A$ and $B$, and then verifying that $\LS{A}\geq \LS{B}$.
Formal construction of the limsup comparator is given here.
Suppose all sequences are natural number sequences, bounded by $\mu$.
The limsup comparator is the B\"uchi automaton $ \LSAut = (\Statess, \Sigma, \delta, \StartState, \Final) $ where,
\begin{itemize}

	\item $ \Statess = \{s\}  \cup \{s_0, s_1 \dots, s_{ \max}\} \cup \{f_0, f_1 \dots, f_{ \max}\}$

	\item $ \Sigma = \{(a,b) : 0 \leq a, b \leq \max \} $ where $ a $ and $ b $ are integers.

	\item $\delta \subseteq \Statess\times\Sigma\times \Statess$ is defined as follows:
	\begin{enumerate}
		\item Transitions from start state $ s $:
		$ (s ,(a,b), p) $ for all $(a,b)\in \Sigma$, and for all $p \in \{s\} \cup \{f_0, f_1, \dots, f_{\max}\}$.

		\item  Transitions between $f_k$ and $s_k$ for each $k$:
		\begin{enumerate}[label = \roman*]
			\item $(f_k, \alpha , f_k)$ for $\alpha \in \{k\} \times \{0,1, \dots k\}$.
			\item $(f_k, \alpha , s_k)$ for $\alpha \in \{0,1,\dots k-1\} \times \{0,1, \dots k\}$.
			\item $(s_k, \alpha , s_k)$ for $\alpha \in \{0,1,\dots k-1\} \times \{0,1, \dots k\}$.
			\item $(s_k, \alpha , f_k)$ for $\alpha \in \{k\} \times \{0,1, \dots k\}$.
		\end{enumerate}

	\end{enumerate}

	\item $ \StartState = \{s\} $

	\item $ \Final =  \{f_0, f_1 \dots, f_{ \max}\}$

\end{itemize}

\noindent
Fig.~\ref{Fig:LSk} illustrates the basic building block of the limsup comparator for relation $\geq$.  For $k \in \{0\dots, \mu\}$, Fig~\ref{Fig:LSk} represents the segment of the limsup comparator consisting of the initial state $s$, accepting state $f_k$ and state $s_k$. We denote this by B\"uchi automaton $\A_k$. We show that automaton $\A_k$ accepts pair $ (A,B)$ of number sequences  iff  $\LS{A} = k$, and $\LS{B}\leq k$, for integer $k$.
 (Lemma~\ref{Lemma:LSFinalState}).

\begin{figure}
	\centering

	\begin{tikzpicture}[shorten >=1pt,node distance=2.5cm,on grid,auto]
	\node[state,initial] (q_0)   {\footnotesize{$s$}};
	\node[state, accepting] (q_1)  [right of = q_0] {\footnotesize{$f_k$}};
	\node[state] (q_2) [right of=q_1] {\footnotesize{$s_k$}};
	\path[->]
	(q_0)   edge [loop above] node {\footnotesize{$(*,*)$}} ()
	edge node  {\footnotesize{$ (k, \leq k) $}} (q_1)
	(q_1)  edge [loop above] node {\footnotesize{$ (k,\leq k) $}} ()
	edge [bend left = 10]  node {\footnotesize{$(\leq k-1, \leq k)$}} (q_2)
	(q_2)  edge [loop above] node {\footnotesize{$ (\leq k-1,\leq k) $}} ()
	edge [bend left = 10]  node {\footnotesize{$(k, \leq k)$}} (q_1);
	\end{tikzpicture}

	\caption{{State $f_k$ is an accepting state.  Automaton $\A_k$ accepts $(A,B)$ iff $\LS{A}=k$, $\LS{B}\leq k$. $*$ denotes $\{0,1\dots \mu\}$, $\leq m$ denotes $\{0,1\dots, m\}$}}%
	\label{Fig:LSk}

\end{figure}
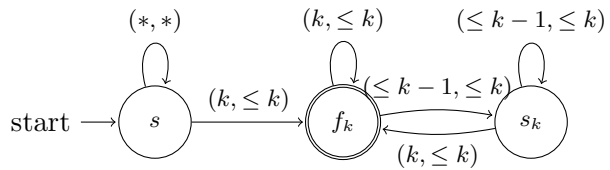

\begin{lem}%
	\label{Lemma:LSFinalState}
	Let $A$ and $B$ be non-negative integer sequences bounded by $\max$.
	B\"uchi automaton $\A_k$ (Fig.~\ref{Fig:LSk}) accepts $(A,B)$ iff $\LS{A} = k$, and $\LS{A}\geq\LS{B}$.
\end{lem}

\begin{proof}
Let $(A,B)$ have an accepting run in $\A_k$. We show that $\LS{A} =k \geq \LS{B}$.
The accepting run visits state $f_k$ infinitely often.
Note that all incoming transitions to accepting state $f_k$ occur on alphabet $(k,\leq k)$ while  all transitions between states $f_k$ and $s_k$ occur on alphabet $(\leq k-1, \leq k)$, where $\leq k$ denotes the set $\{0,1,\dots k\}$.
So, the integer $k$ must appear infinitely often in $A$ and all elements occurring infinitely often in $A$ and $B$ are less than or equal to $k$. Therefore, if $(A,B)$ is accepted by $\A_k$ then $\LS{A} =k$, and $\LS{B} \leq k$, and $\LS{A}\geq \LS{B}$.

Conversely, let $\LS{A} = k > \LS{B}$. We prove that $(A,B)$ is accepted by $\A_k$.
	For an integer sequence $A$ when $ \LS{A} = k $ integers greater than $k$ can occur only a finite number of times in $A$. Let $l_A$ denote the index of the last occurrence of an integer greater than $k$ in $A$.
	Similarly, since $\LS{B} \leq k$, let $l_B$ be index of the last occurrence of an integer greater than $k$ in $B$.
	Therefore, for sequences $A$ and $B$ integers greater than $k$ will not occur  beyond index $l = \mathit{max}(l_A, l_B)$.
	B\"uchi automaton $\A_k$ (Fig.~\ref{Fig:LSk}) non-deterministically determines $l$.
	On reading the $l$-th element of input word $(A,B)$, the run of $(A,B)$ exits the start state $s$ and  shifts to accepting state $f_k$. Note that all runs beginning at state $f_k$ occur on alphabet $(a,b)$ where $a,b\leq k$. Therefore, $(A,B)$ can continue its infinite run even after transitioning to $f_k$. To ensure that this is an accepting run, the run must visit accepting state $f_k$ infinitely often. But this must be the case, as transition on alphabet $(k,k')$ for $k'\leq k$ must be taken infinitely often as $k$ is the limsup of $A$ and limsup of $B$ is less than or equal to  $k$. Transitions on this alphabet always return to the accepting state $f_k$.
	Hence, for all integer sequences $A$,$B$ bounded by $\max$, if $\LS{A}=k$, and $\LS{A}\geq \LS{B}$, the automaton accepts $(A,B)$.
\end{proof}
\begin{thm}%
	\label{thrm:limsupcomparator}
    There exists an $\omega$-regular comparator for the limsup aggregation function.
\end{thm}
\begin{proof}
	The construction given above  contains a B\"uchi automata $\A_k$ for all $k \in \{0, \dots, \mu\}$. Therefore, from Lemma~\ref{Lemma:LSFinalState}, we conclude that the construction corresponds to the limsup comparator with inequality $\geq$. Therefore, limsup comparator with relation $\geq$ is $\omega$-regular.

	From  Lemma~\ref{lemm:allregular} we know that limsup comparator is $\omega$-regular for all relations.
%
\end{proof}
Due to closure properties of B\"uchi automata, this implies that limsup comparator for all inequalities and equality relation is also $\omega$-regular.
The \emph{limit infimum} (liminf, in short) of an integer sequence is the smallest integer that appears infinitely often in it; its comparator has a similar construction to the limsup comparator.
One can further prove that the limsup and liminf aggregate functions are also $\omega$-regular aggregate functions.

\subsection{\texorpdfstring{$\omega$}{Omega}-Regular aggregate functions}%
\label{Sec:RegularAggFunctions}

This section draws out the relationship between $\omega$-regular aggregate functions and $\omega$-regular comparators. We begin with the following Lemma in order to  show that $\omega$-regular aggregate functions entail $\omega$-regular comparators for the aggregate function.
\begin{lem}%
\label{lemma:reviewerAdd}
Let $\mu>0$ be the upper-bound on weight sequences, and $\beta\geq 2$ be the integer base.
Then
there exists a B\"uchi automaton $\A_\beta$ such that for all $a,b\in \Re$, $\A_\beta$ accepts $(\mathsf{rep}(a,\beta), \mathsf{rep}(b, \beta))$ iff $a>b$.
\end{lem}

\begin{proof}

Let $a,b\in\Re$, and $\beta>2$ be an integer base. Let $\mathsf{rep}(a,\beta) = \mathsf{sign}_a\cdot(\mathsf{Int}(a, \beta), \mathsf{Frac}(a, \beta))$ and $ \mathsf{rep}(b,\beta)  = \mathsf{sign}_b\cdot(\mathsf{Int}(b, \beta), \mathsf{Frac}(b, \beta))$. Then, the following statements can be proven using simple evaluation from definitions:
\begin{itemize}
\item When $\mathsf{sign}_a =+$ and $\mathsf{sign}_b = -$. Then $a > b$.
\item When $\mathsf{sign}_a = \mathsf{sign}_b = +$
	\begin{itemize}
	\item  If $\mathsf{Int}(a, \beta) \neq \mathsf{Int}(b, \beta)$: Since $\mathsf{Int}(a, \beta)$ and $\mathsf{Int}(b, \beta)$ eventually only see digit $0$ i.e.\ they are necessarily identical eventually, there exists an index $i$ such that it is the last position where $\mathsf{Int}(a, \beta)$ and $\mathsf{Int}(b, \beta)$ differ.
	If $\mathsf{Int}(a, \beta)[i]>\mathsf{Int}(b, \beta)[i]$, then $a > b$.
	If $\mathsf{Int}(a, \beta)[i]<\mathsf{Int}(b, \beta)[i]$, then $a < b$.

	\item If $\mathsf{Int}(a, \beta) = \mathsf{Int}(b, \beta)$ but  $\mathsf{Frac}(a, \beta) \neq \mathsf{Frac}(b, \beta)$:
	Let $i$ be the first index where $\mathsf{Frac}(a, \beta)$ and $\mathsf{Frac}(b, \beta)$ differ.
	If $\mathsf{Frac}(a, \beta)[i] > \mathsf{Frac}(b, \beta)[i]$ then $a > b$.
	If $\mathsf{Frac}(a, \beta)[i] < \mathsf{Frac}(b, \beta)[i]$ then $a < b$.

	\item Finally, if $\mathsf{Int}(a, \beta) = \mathsf{Int}(b, \beta)$ and $\mathsf{Frac}(a, \beta) = \mathsf{Frac}(b, \beta)$: Then $a = b$.
	\end{itemize}
\item When $\mathsf{sign}_a = \mathsf{sign}_b = -$
		\begin{itemize}
		\item  If $\mathsf{Int}(a, \beta) \neq \mathsf{Int}(b, \beta)$: Since $\mathsf{Int}(a, \beta)$ and $\mathsf{Int}(b, \beta)$ eventually only see digit $0$ i.e.\ they are necessarily identical eventually. Therefore, there exists an index $i$ such that it is the last position where $\mathsf{Int}(a, \beta)$ and $\mathsf{Int}(b, \beta)$ differ.
		If $\mathsf{Int}(a, \beta)[i]>\mathsf{Int}(b, \beta)[i]$, then $a < b$.
		If $\mathsf{Int}(a, \beta)[i]<\mathsf{Int}(b, \beta)[i]$, then $a > b$.

		\item If $\mathsf{Int}(a, \beta) = \mathsf{Int}(b, \beta)$ but  $\mathsf{Frac}(a, \beta) \neq \mathsf{Frac}(b, \beta)$:
		Let $i$ be the first index where $\mathsf{Frac}(a, \beta)$ and $\mathsf{Frac}(b, \beta)$ differ.
		If $\mathsf{Frac}(a, \beta)[i] > \mathsf{Frac}(b, \beta)[i]$ then $a < b$.
		If $\mathsf{Frac}(a, \beta)[i] < \mathsf{Frac}(b, \beta)[i]$ then $a > b$.

		\item Finally, if $\mathsf{Int}(a, \beta) = \mathsf{Int}(b, \beta)$ and $\mathsf{Frac}(a, \beta) = \mathsf{Frac}(b, \beta)$: Then $a = b$.
	\end{itemize}
\item When $\mathsf{sign}_a =-$ and $\mathsf{sign}_b = +$. Then $a < b$.
\end{itemize}
Since the conditions given above are exhaustive and mutually exclusive, we conclude that for all $a,b\in\Re$ and integer base $\beta\geq 2$, letting $\mathsf{rep}(a,\beta) = \mathsf{sign}_a\cdot(\mathsf{Int}(a, \beta), \mathsf{Frac}(a, \beta))$ and $ \mathsf{rep}(b,\beta)  = \mathsf{sign}_b\cdot(\mathsf{Int}(b, \beta), \mathsf{Frac}(b, \beta))$,  $a>b$ iff one of the following conditions occurs:
\begin{enumerate}
	\item $\mathsf{sign}_a =+$ and $\mathsf{sign}_b = -$.
	\item $\mathsf{sign}_a = \mathsf{sign}_b = +$, $\mathsf{Int}(a, \beta) \neq \mathsf{Int}(b, \beta)$, and $\mathsf{Int}(a, \beta)[i]>\mathsf{Int}(b, \beta)[i]$ when $i$ is the last index where $\mathsf{Int}(a, \beta)$ and $\mathsf{Int}(b, \beta)$ differ.
	\item $\mathsf{sign}_a = \mathsf{sign}_b = +$, $\mathsf{Int}(a, \beta) = \mathsf{Int}(b, \beta)$, $\mathsf{Frac}(a, \beta) \neq \mathsf{Frac}(b, \beta)$, and $\mathsf{Int}(a, \beta)[i]>\mathsf{Int}(b, \beta)[i]$ when $i$ is the first index where $\mathsf{Frac}(a, \beta)$ and $\mathsf{Frac}(b, \beta)$ differ.
	\item $\mathsf{sign}_a = \mathsf{sign}_b = -$, $\mathsf{Int}(a, \beta) \neq \mathsf{Int}(b, \beta)$, and $\mathsf{Int}(a, \beta)[i]<\mathsf{Int}(b, \beta)[i]$ when $i$ is the last index where $\mathsf{Int}(a, \beta)$ and $\mathsf{Int}(b, \beta)$ differ.
	\item $\mathsf{sign}_a = \mathsf{sign}_b = -$, $\mathsf{Int}(a, \beta) = \mathsf{Int}(b, \beta)$, $\mathsf{Frac}(a, \beta) \neq \mathsf{Frac}(b, \beta)$, and $\mathsf{Int}(a, \beta)[i]<\mathsf{Int}(b, \beta)[i]$ when $i$ is the first index where $\mathsf{Frac}(a, \beta)$ and $\mathsf{Frac}(b, \beta)$ differ.
\end{enumerate}
Note that each of these five conditions can be easily expressed by a B\"uchi automaton over alphabet $\mathsf{AlphaRep(\beta)}$ for an integer $\beta\geq 2$. For an integer $\beta\geq 2$, the union of all these B\"uchi automata will result in a B\"uchi automaton $\A_\beta$ such that for all $a,b\in \Re$ and $A = \mathsf{rep}(a, \beta)$ and $B = \mathsf{rep}(b,\beta)$, $a>b$ iff interleaved word $(A,B) \in \L(\A_\beta)$.
\end{proof}
We finally  show that $\omega$-regular aggregate functions entail $\omega$-regular comparators for the aggregate function.
\begin{thm}%
\label{thrm:functionthencomparator}
Let $\mu>0$ be the upper-bound on weight sequences, and $\beta\geq 2$ be the integer base. Let $f: \{0,1,\dots,\mu\}^\omega\rightarrow \Re$ be an aggregate function. If aggregate function $f$ is $\omega$-regular under base $\beta$, then its comparator for all inequality and equality relations is also $\omega$-regular.
\end{thm}
\begin{proof}
We show that if an aggregate function is $\omega$-regular under base $\beta$, then its comparator for relation $>$ is $\omega$-regular. By closure properties of $\omega$-regular comparators, this implies that comparators of the aggregate function are $\omega$-regular for all inequality and equality relations.

Let $f:\Sigma^\omega\rightarrow \Re$ be an $\omega$-regular aggregate function with aggregate function automata $\A_f$. We will construct an $\omega$-regular comparator for $f$ with relation $>$. From Lemma~\ref{lemma:reviewerAdd} we know that $(X,Y)$ is present in the comparator iff $(X, M), (Y,N)\in \A_f$ for $M,N \in \mathsf{AlphaRep}(\beta)^\omega$ and $(M,N) \in \A_\beta$, for $\A_\beta$ as described above. Since $\A_f$ and $\A_\beta$ are both B\"uchi automata, the comparator for function $f$ with relation $>$ is also a B\"uchi auotmaton. Therefore, the comparator for aggregate function $f$ with relation $> $ is $\omega$-regular.
\end{proof}
The converse direction of whether $\omega$-regular comparator for an aggregate function $f$ for all inequality or equality relations will entail $\omega$-regular functions under an integer base $\beta\geq 0$  is trickier. For all aggregate functions considered in this paper, we see that whenever the comparator is $\omega$-regular, the aggregate function is $\omega$-regular as well.
However, the proofs for this have been done on a case-by-cass basis, and we do not have an algorithmic procedure to derive a function (B\"uchi) automaton from its $\omega$-regular comparator. We also do not have an example of an aggregate function for which the comparator is $\omega$-regular but the function is not. Therefore, we arrive at the following  conjecture:
\begin{conj}%
\label{Conjecture:comparatortofunction}
Let $\mu>0$ be the upper-bound on weight sequences, and $\beta\geq 2$ be the integer base. Let $f: \{0,1,\dots,\mu\}^\omega\rightarrow \Re$ be an aggregate function. If the comparator for an aggregate function $f$ is $\omega$-regular for all inequality and equality relations, then its aggregate function is also $\omega$-regular under base $\beta$.
\end{conj}

\subsection{Quantitative inclusion}%
\label{Sec:Inclusion}
{The aggregate function or comparator of a quantitative inclusion} problem refer to the aggregate function or comparator  of the associated aggregate function.
This section presents a generic algorithm (Algorithm~\ref{Alg:DSInclusion})  to solve quantitative inlcusion between $\omega$-weighted automata $P$ and $Q$ with $\omega$-comparators.
This section focusses on the non-strict quantitative inclusion.
$\DSInclusion$ (Algorithm~\ref{Alg:DSInclusion}) is an algorithm for quantitative inclusion between  weighted $\omega$-automata $P$ and $Q$ with  $\omega$-regular comparator $\A_f$ for relation $\geq$.  $\DSInclusion$ takes $P$,$Q$ and $\A_f$ as input, and returns $\mathsf{True}$ iff $P \subseteq_f Q$.
The results for strict quantitative inclusion are similar.
We use the following motivating example to explain steps of Algorithm~\ref{Alg:DSInclusion}.

\paragraph{Motivating example}%
\label{Sec:Example}
\begin{figure*}[t]
	\centering
	\begin{minipage}{0.48\textwidth}
		\centering
		\begin{tikzpicture}[shorten >=1pt,node distance=2.5cm,on grid,auto]
		\node[state,accepting, initial] (q_1)   {$p_1$};
		\node[state, accepting] (q_2) [right of = q_1] {$p_2$};

		\path[->]
		(q_1)	edge  node  {$a,1$} (q_2)
		(q_2)   edge [loop right] node  {$a,1$} (q_2);

		\end{tikzpicture}
		\caption{Weighted automaton $P$}%
		\label{Fig:WA-P}
	\end{minipage}
	\hfill
	\centering
		\begin{minipage}{0.48\textwidth}
		\centering
		\begin{tikzpicture}[shorten >=1pt,node distance=2.5cm,on grid,auto]

		\node[state,accepting, initial] (q_1)   {$q_1$};
		\node[state, accepting] (q_2) [right of = q_1] {$q_2$};

		\path[->]
		(q_1)	edge  [bend left] node  {$a,0$} (q_2)
		edge  node  [below] {$a,2$} (q_2)
		(q_2)   edge [loop right] node  {$a,1$} (q_2);

		\end{tikzpicture}
		\caption{Weighted automaton $Q$}%
		\label{Fig:WA-Q}
	\end{minipage}
	\hfill
\end{figure*}
Let weighted $\omega$-automata $P$ and $Q$ be as illustrated in Fig.~\ref{Fig:WA-P}-\ref{Fig:WA-Q} with the limsup aggregate function.
The word $w=a^{\omega}$ has one run $\rho^P_1 = p_1p_2^{\omega}$ with weight sequence $wt^P_1=1^{\omega}$ in $P$ and two runs $\rho^Q_1 = q_1q_2^{\omega}$ with weight sequence $wt^Q_1=0,1^{\omega}$ and run $\rho^Q_2 = q_1q_2^{\omega}$ with weight sequence $wt^Q_2=2,1^{\omega}$. Clearly, $wt_P(w) \leq wt_Q(w$). Therefore $P\subseteq_f Q$. From Theorem~\ref{thrm:limsupcomparator} we know that the limsup comparator $\A_{\mathsf{LS}}^\leq$ for $\leq$ is $\omega$-regular.

We use Algorithm~\ref{Alg:DSInclusion} to show that $P \subseteq_f Q$ using its $\omega$-regular comparator for $\leq$. Intuitively, the algorithm must be able to identify that for run $\rho^P_1$ of $w$ in $P$, there exists a run $\rho^Q_2$ in $Q$ s.t. $(wt^P_1, wt^Q_2)$ is accepted by the limsup comparator for $\leq$.

\paragraph{Key ideas}

A run $\rho_P$ in $P$ on word $w\in \Sigma^\omega$ is said to be \emph{dominated} w.r.t $P\subseteq_f Q$ if there exists a run $\rho_Q$ in $Q$ on the same word $w$ such that $wt_P(\rho_P)\leq wt_Q(\rho_Q)$.
$P \subseteq_f Q$ holds if for every run $\rho_P$ in $P $ is dominated w.r.t. $P\subseteq_f Q$.

$\DSInclusion$ constructs B\"uchi automaton $\first$ that consists of exactly the dominated runs of $P$ w.r.t $P\subseteq_f Q$. $\DSInclusion$ returns $\mathsf{True}$ iff $\first$ contains all runs of $P$.
To obtain $\first$, it constructs B\"uchi automaton $\less$ that accepts  word $(\rho_P, \rho_Q)$ iff $\rho_P$ and $\rho_Q$ are runs of the same word in  $P$ and $Q$ respectively, and $wt_P(\rho_P)\leq wt_Q(\rho_Q)$ i.e.\ if $w_P$ and $w_Q$ are weight sequence of $\rho_P$ and $\rho_Q$, respectively, then $(w_P, w_Q)$ is present in the $\omega$-regular comparator $\A_f^\leq$ for aggregate function $f$ with relation $\leq$.
The projection of $\less$ on runs of $P$ results in $\first$.

\begin{algorithm}[t]
\caption{ $\DSInclusion(P,Q, \A_f)$, Is $P \subseteq_f Q$?}%
\label{Alg:DSInclusion}
\begin{algorithmic}[1]
\STATE \textbf{Input: } Weighted automata $P$, $Q$, and $\omega$-regular comparator $\A_f$ (Inequality $\leq$)
\STATE \textbf{Output: } $\mathsf{True}$ if $P \subseteq_f Q$, $\mathsf{False}$ otherwise

\STATE $\hat{P} \leftarrow  \AugmentWtAndLabel(P)$\label{alg-line:AugmentP}
\STATE $\hat{Q} \leftarrow  \AugmentWtAndLabel(Q)$\label{alg-line:AugmentQ}

 \STATE $\hat{P}\times\hat{Q} \leftarrow \MakeProduct(\hat{P}, \hat{Q})$\label{alg-line:Prod}

\STATE $\less \leftarrow  \Intersect(\hat{P}\times\hat{Q}, \A_{\succeq} )$\label{alg-line:Intersect}

\STATE $\first \leftarrow \Project(\less)$\label{alg-line:Project}
\RETURN {$\hat{P} \equiv \first$}\label{alg-line:ensure}

\end{algorithmic}
\end{algorithm}
\paragraph{Algorithm details}
For sake a simplicity, we assume that every word present in $P$ is also present in $Q$ i.e. $P\subseteq Q$ (qualitative inclusion).
$\DSInclusion$ has three steps:
(a). $\SepRuns$ (Lines~\ref{alg-line:AugmentP}-\ref{alg-line:AugmentQ}):  Enables unique identification of runs in $P$ and $Q$ through \emph{labels}.
(b). $\CompareRuns$ (Lines~\ref{alg-line:Prod}-\ref{alg-line:Project}): Compares weight of runs in $P$ with weight of runs in $Q$, and constructs $\first$. (c). $\EnsureRuns$ (Line~\ref{alg-line:ensure}): Ensures if all runs of $P$ are diminished.

\begin{enumerate}
\item $\SepRuns$: $\AugmentWtAndLabel$ transforms weighted $\omega$-automaton $\A$ into B\"uchi automaton $\hat{\A}$ by converting transition $\tau = (s, a, t )$ with weight $\gamma(\tau)$ in $ \A$ to transition $\hat{\tau} =  (s, (a, \gamma(\tau), l), t)$ in $\hat{\A}$, where $l$ is  a unique label assigned to transition $\tau$. The word $\hat{\rho} = (a_0,n_0,l_0)(a_1,n_1,l_1)\dots \in \hat{A}$  iff  there exists a run $\rho \in \A$ on word $a_0a_1\dots$ with weight sequence $n_0n_1\dots$.
Labels ensure bijection between runs in $\A$ and words in $\hat{\A}$.
Words of $\hat{A}$ have a single run in $\hat{A}$.
Hence, transformation of weighted $\omega$-automata  $P$ and $Q$ to B\"uchi automata $\hat{P}$ and $\hat{Q}$ enables disambiguation between runs of $P$ and $Q$ (Line~\ref{alg-line:AugmentP}-\ref{alg-line:AugmentQ}).

The corresponding $\hat{A}$ for weighted $\omega$-automata $P$ and $Q$ from Figure~\ref{Fig:WA-P}-~\ref{Fig:WA-Q} are given in Figure~\ref{Fig:WA-Phat}-~\ref{Fig:WA-Qhat} respectively.

\begin{figure*}[t]
	\centering
	\begin{minipage}{0.30\textwidth}
		\centering
        \scalebox{0.8}{%
            \begin{tikzpicture}[shorten >=1pt,node distance=2cm,on grid,auto]
            \node[state,accepting, initial] (q_1)   {$p_1$};
            \node[state, accepting] (q_2) [right of = q_1] {$p_2$};

            \path[->]
            (q_1)	edge  node  {$(a,1,1)$} (q_2)
            (q_2)   edge [loop above] node  {$(a,1,2)$} (q_2);

            \end{tikzpicture}
        }
		\caption{$\hat{P}$}%
		\label{Fig:WA-Phat}
	\end{minipage}
	\hfill
	\centering
	\begin{minipage}{0.28\textwidth}
		\centering
        \scalebox{0.8}{%
            \begin{tikzpicture}[shorten >=1pt,node distance=2cm,on grid,auto]

            \node[state,accepting, initial] (q_1)   {$q_1$};
            \node[state, accepting] (q_2) [right of = q_1] {$q_2$};

            \path[->]
            (q_1)	edge  [bend left] node  {$(a,0,1)$} (q_2)
            edge   node  [below] {$(a,2,2)$} (q_2)
            (q_2)   edge [loop above] node  {$(a,1,3)$} (q_2);

            \end{tikzpicture}
        }
		\caption{$\hat{Q}$}%
		\label{Fig:WA-Qhat}
	\end{minipage}
	\hfill
	\begin{minipage}{0.4\textwidth}
		\centering
        \scalebox{0.8}{%
            \begin{tikzpicture}[shorten >=1pt,node distance=3.1cm,on grid,auto]
            \node[state,accepting, initial] (s_1)   {$p_1,q_1$};
            \node[state, accepting] (s_2) [right of = s_1] {$p_2,q_2$};

            \path[->]
            (s_1)	edge  [bend left] node  {$(a,1,1,0,1)$} (s_2)
            (s_1)	edge  node  [below] {$(a,1,1,2,2)$} (s_2)
            (s_2)   edge [loop, distance=1.5cm] node  [above] {$(a,1,2,1,3)$} (s_2);

            \end{tikzpicture}
        }
		\caption{$\hat{P}\times\hat{Q}$}%
		\label{Fig:PhatQhat}
	\end{minipage}
\end{figure*}

\item $\CompareRuns$:  The output of this step is the B\"uchi automaton $\first$ that contains the word ${\hat{\rho}}\in\hat{P}$ iff $\rho$ is a dominated run in $P$ w.r.t $P\subseteq_f Q$ (Lines~\ref{alg-line:Prod}-\ref{alg-line:Project}).

$\MakeProduct(\hat{P}, \hat{Q})$ constructs $\hat{P}\times \hat{Q}$ s.t.\ word $(\hat{\rho_P}, \hat{\rho_Q}) \in \hat{P}\times \hat{Q}$ iff $\rho_P$ and $\rho_Q$ are runs of the same word in $P$ and $Q$ respectively (Line~\ref{alg-line:Prod}).
Concretely, for transition $\hat{\tau_\A}=(s_\A,(a, n_\A, l_\A),t_\A)$
in automaton $\A$, where $\A \in \{\hat{P}, \hat{Q}\}$,
 transition $\hat{\tau_P}\times\hat{\tau_Q}=((s_P, s_Q),(a, n_P, l_P, n_Q, l_Q),(t_P, t_Q))$ is in $ \hat{P}\times\hat{Q}$, as shown in Figure~\ref{Fig:PhatQhat}.


$\Intersect$ intersects the weight components of $\hat{P}\times\hat{Q}$ with comparator $\A_f^\leq$ (Line~\ref{alg-line:Intersect}). The resulting  automaton $\less$ accepts word $(\hat{\rho_P}, \hat{\rho_Q})$ iff $f(\rho_P)\leq f(\rho_Q)$, and $\rho_P$ and $\rho_Q$ are runs on the same word in $P$ and $Q$ respectively. The result of $\Intersect$ between $\hat{P}\times\hat{Q}$ with the limsup comparator $\A_{\mathsf{LS}}^\leq$ for relation $\leq$ (Figure~\ref{Fig:SupComparator}) is given in Figure~\ref{Fig:Intersect}.

The projection of $\less$ on the alphabet of $\hat{P}$ returns $\first$ which contains the word $\hat{\rho_P}$ iff $\rho_P$ is a dominated run in $P$ w.r.t $P\subseteq_f Q$ (Line~\ref{alg-line:Project}), as shown in Figure~\ref{Fig:Dim}.

\begin{figure*}[t]
	\centering
	\begin{minipage}[t]{0.3\textwidth}
		\centering
        \scalebox{0.8}{%
            \begin{tikzpicture}[shorten >=1pt,node distance=1.7cm,on grid,auto]
            \node[state,initial] (s_1)   {$s_1$};
            \node[state, accepting] (s_2) [right of = s_1] {$s_2$};
            \node[state] (s_3) [below right of = s_1] {$s_3$};

            \path[->]
            (s_1)	edge   node  {$(1,2)$} (s_2)
            (s_1)	edge  node  [below,sloped] {$(1,0)$} (s_3)
            (s_2)   edge [loop right] node  [above] {$(1,1)$} (s_2)
            (s_3)   edge [loop right] node  [above]  {$(1,1)$} (s_3);

            \end{tikzpicture}
        }
		\caption{Snippet of limsup comparator $\A_{\mathsf{LS}}^\leq$ for relation $\leq$}%
		\label{Fig:SupComparator}
	\end{minipage}
	\hfill
	\centering
		\begin{minipage}[t]{0.34\textwidth}
		\centering
        \scalebox{0.8}{%
            \begin{tikzpicture}[shorten >=1pt,node distance=2.9cm,on grid,auto]
            \node[state,accepting, initial] (s_1)   {$t_1$};
            \node[state, accepting] (s_2) [right of = s_1] {$t_2$};

            \path[->]
            (s_1)	edge  node  [below] {$(a,1,1,2,2)$} (s_2)
            (s_2)   edge [loop, distance=1.5cm] node  [above] {$(a,1,2,1,3)$} (s_2);

            \end{tikzpicture}
        }
		\caption{Snippet of $\Intersect$. $t_1 = (p_1,q_1,s_1)$ and $t_2 = (p_2,q_2,s_2)$.}%
		\label{Fig:Intersect}
	\end{minipage}
	\hfill
	\begin{minipage}[t]{0.28\textwidth}
		\centering
        \scalebox{0.8}{%
            \begin{tikzpicture}[shorten >=1pt,node distance=2.1cm,on grid,auto]
            \node[state,accepting, initial] (q_1)   {$t_1'$};
            \node[state, accepting] (q_2) [right of = q_1] {$t_2'$};

            \path[->]
            (q_1)	edge  node  {$(a,1,1)$} (q_2)
            (q_2)   edge [loop above] node  {$(a,1,2)$} (q_2);

            \end{tikzpicture}
        }
		\caption{Snippet of $\first$}%
		\label{Fig:Dim}
	\end{minipage}
\end{figure*}

\item $\EnsureRuns$: $P\subseteq_f Q$ iff $\hat{P} \equiv \first$ (qualitative equivalence) since $\hat{P}$ consists of all runs of $P$ and $\first$ consists of all dominated runs w.r.t $P\subseteq_f Q$ (Line~\ref{alg-line:ensure}).
\end{enumerate}

\begin{lem}%
	\label{lem:dim}
	B\"uchi automaton $\first$ consists of all dominated runs in $P$ w.r.t $P\subseteq_f Q$.
\end{lem}

\begin{proof}
	Let $\A_f^\leq$ be the comparator for $\omega$-regular aggregate function $f$ and relation $\leq$ s.t. $\A_f$ accepts $(A,B)$ iff $f(A)\leq f(B)$.
	A run $\rho$ over word $w$ with weight sequence $wt$ in $P$ (or $Q$) is represented by the unique word $\hat{\rho} = (w, wt, l)$ in $\hat{P}$ (or $\hat{Q}$) where $l$ is the unique label sequence associated with each run in $P$ (or $Q$). Since every label on each transition is unique, $\hat{P}$ and $\hat{Q}$ are deterministic automata. Now, $\hat{P}\times\hat{Q}$ is constructed by ensuring that two transitions are combined in the product only if their alphabet is the same.
	Therefore if $(w, wt_1, l_1, wt_2, l_2) \in \hat{P}\times\hat{Q}$, then $\hat{\rho} = (w, wt_1, l_1)\in \hat{P}$, $\hat{\sigma} = (w, wt_2, l_2)\in \hat{Q}$. Hence, there exist runs $\rho$ and $\sigma$ with weight sequences $wt_1$ and $wt_2$ in $P$ and $Q$, respectively.
	Next, $\hat{P}\times\hat{Q}$ is intersected over the weight sequences with $\omega$-regular comparator $\A_f^\leq$ for aggregate function $f$ and relation $\leq$.
	Therefore $(w, wt_1, l_1, wt_2, l_2) \in \less$
	iff  $f(wt_1) \leq f(wt_2)$. Therefore runs $\rho$ in $P$ and $\sigma$ in $Q$ are runs on the same word s.t.\ aggregate weight in $P$ is less than or equal to that of $\sigma$ in $Q$. Therefore $\first$ consists of $\hat{\rho}$ only if $\rho$ is a dominated run in $P$ w.r.t $P\subseteq_f Q$.

	Every step of the algorithm has a two-way implication, hence it is also true that every dominated run in $P$ w.r.t $P\subseteq_f Q$ is present in $\first$.
\end{proof}

\begin{lem}%
\label{Lemma:DSInclusionAlg}
Given weighted $\omega$-automata $P$ and $Q$ and their  $\omega$-regular  comparator $\A_f^\leq$ for aggregate function $f$ and relation $\leq$.
$\DSInclusion(P,Q,\A_f)$ returns $\mathsf{True}$ iff $P \subseteq_f Q$.
\end{lem}
\begin{proof}
	$\hat{P}$ consists of all runs of $P$. $\first$ consists of all dominated run in $P$ w.r.t $P\subseteq_f Q$. $P\subseteq_f Q$ iff every run of $P$ is dominated w.r.t $P\subseteq_f Q$. Therefore $P\subseteq_f Q$ is given by whether $\hat{P} \equiv \first$, where $\equiv$ denotes qualitative equivalence.
\end{proof}
Algorithm $\DSInclusion$ is adapted for \emph{strict} quantitative inclusion $P\subset_f Q$  by repeating the same procedure with $\omega$-regular comparator $\A_f^<$ for aggregate function $ f$ and relation $<$. Here, a run $\rho_P$ in $P$ on word $w\in \Sigma^\omega$ is said to be \emph{dominated} w.r.t $P\subset_f Q$ if there exists a run $\rho_Q$ in $Q$ on the same word $w$ such that $wt_P(\rho_P)< wt_Q(\rho_Q)$. Similarly for quantitative equivalence $P \equiv_f Q$.

We give the complexity analysis of quantitative-inclusion with $\omega$-regular comparators.

\begin{thm}%
	\label{thrm:RegularComplexity}

	Let $P$ and $Q$ be weighted $\omega$-automata and $\A_f$ be an $\omega$-regular comparator.
	Quantitative inclusion problem, quantitative strict-inclusion problem, and quantitative equivalence problem for $\omega$-regular aggregate function $f$ is $\cc{PSPACE}$-complete.

\end{thm}
\begin{proof}
All operations in $\DSInclusion$ until Line~\ref{alg-line:Project} are polytime operations in the  size of weighted $\omega$-automata $P$, $Q$ and comparator $\A_f$. Hence, $\first$ is polynomial in size of $P$, $Q$ and $\A_f$. Line~\ref{alg-line:ensure} solves a $\cc{PSPACE}$-complete problem.
Therefore, the quantitative inclusion for $\omega$-regular aggregate function $f$ is in $\cc{PSPACE}$ in size of the inputs $P$, $Q$, and  $\A_f$.

The $\cc{PSPACE}$-hardness of the quantitative inclusion is established via reduction from the \emph{qualitative} inclusion problem, which is   $\cc{PSPACE}$-complete.
The formal reduction is as follows: Let $P$ and $Q$ be B\"uchi automata (with all states as accepting states). Reduce $P$, $ Q$ to weighted automata $\overline{P}$, $\overline{Q}$ by assigning a weight of 1 to each transition.
Since all runs in $\overline{P}$, $\overline{Q}$ have the same weight sequence, weight of all words in $\overline{P}$ and $\overline{Q}$ is the same for any function $f$.
It is easy to see $P \subseteq Q$ (qualitative inclusion) iff $\overline{P} \subseteq_f \overline{Q}$ (quantitative inclusion).
\end{proof}
Theorem~\ref{thrm:RegularComplexity} extends to weighted $\omega$-automata when weight of words is the \emph{infimum} of weight of runs.
The key idea for $P\subseteq_f Q$ here is to ensure that for every run $\rho_Q$ in $Q$ there exists a run  on the same word in $\rho_P$ in $P$  s.t. $f(\rho_P)\leq f(\rho_Q)$.

\label{Sec:Represent} 
\paragraph{Representation of counterexamples}

When $P \nsubseteq_f Q$,
there exists word(s) $w \in \Sigma^*$ s.t $wt_P(w)>wt_Q(w)$.
Such a word $w$ is said to be a {\em counterexample word}.
Previously, finite-state representations of counterexamples have been useful in verification and synthesis in qualitative systems~\cite{baier2008principles}, and could be useful in quantitative settings as well.
However, we are not aware of procedures for such representations in the quantitative settings.
Here we show that a trivial extension of $\DSInclusion$ yields B\"uchi automata-representations for all counterexamples of the quantitative inclusion problem for $\omega$-regular functions.

\begin{thm}%
	\label{Thm:RegularInclusion-Counterexample}
All counterexamples of the quantitative inclusion problem for an $\omega$-regular aggregate function can be expressed by a B\"uchi automaton.

\end{thm}
\begin{proof}
For word $w$ to be a counterexample, it must contain a run in $P$ that is not dominated. Clearly, all non-dominated runs of $P$ w.r.t to the quantitative inclusion are members of $\hat{P}\setminus \first$. The counterexamples words can be obtained from $\hat{P}\setminus \first$ by modifying its alphabet to the alphabet of $P$ by dropping transition weights and their unique labels.
\end{proof}

\subsection{Incomplete-information quantitative games}%
\label{Sec:GraphGames}

Given an incomplete-information quantitative game $\G = (S,s_ {\mathcal{I}},\Obs, \Sigma, \delta, \gamma, f)$, our objective is to determine if player $P_0$ has a winning strategy $ \alpha: \Obs^* \rightarrow \Sigma$ for $\omega$-regular aggregate function $f$. We assume we are given the $\omega$-regular comparator $\A_f$ for function $f$. We provide an informal description of the algorithm to describe the intuition.

Note that a function $A^*\rightarrow B$ can be treated like a  $B$-labeled $A$-tree, and vice-versa.
Hence, we proceed by finding a $\Sigma$-labeled $\Obs$-tree --- the \emph{winning strategy tree}. Every branch of a winning strategy-tree is an {observed play} $o_{\rho}$ of $\G$ for which every actual play $\rho$ is a winning play for $P_0$.

We first consider all \emph{game trees} of $\G$ by interpreting $\G$ as a tree-automaton over $\Sigma$-labeled $S$-trees. Nodes $n \in S^*$ of the game-tree correspond to states in $S$ and labeled by actions in $\Sigma$ taken by player $P_0$. Thus, the \emph{root node} $\varepsilon$ corresponds to $s_\I$, and a node $s_{i_0},\ldots,s_{i_k}$ corresponds to the state $s_{i_k}$ reached via $s_{\I},s_{i_0},\ldots,s_{i_{k-1}}$.
Consider now a node $x$ corresponding to state $s$ and labeled by an action $\sigma$. Then $x$ has children $x s_1,\ldots x s_n$, for every $s_i \in S$. If $s_i\in \delta(s,\sigma)$, then we call $x s_i$ a \emph{valid} child, otherwise we call it an \emph{invalid} child. Branches that contain invalid children correspond to invalid plays.

A game-tree $\tau$ is a \emph{winning tree} for player $P_0$ if every branch of $\tau$ is either a winning play for $P_0$ or an invalid play of $\G$.
One can check, using an automata, if a play is invalid by the presence of invalid children.
Furthermore, the winning condition for $P_0$ can be expressed by the $\omega$-regular comparator $\A_f$ that accepts $(A,B)$ iff $f(A) > f(B)$.
To use the comparator $ \A_f$, it is determinized to  Parity automaton $D_f$.
Thus, a product of  game $\G$ with $D_f$ is a deterministic Parity tree-automaton accepting precisely winning-trees for player $P_0$.

Winning trees for player $P_0$ are $\Sigma$-labeled $S$-trees. We need to convert them to $\Sigma$-labeled $\Obs$-trees. Recall that every state has a unique observation.
We can simulate these $\Sigma$-labeled $S$-trees on strategy trees using the technique of \emph{thinning} states $S$ to observations $\Obs$~\cite{kupferman2000synthesis}. The resulting alternating Parity tree automaton $\mathcal{M}$ will accept a $\Sigma$-labeled $\Obs$-tree $\tau_o$ iff for all actual game-tree $\tau$ of $\tau_o$, $\tau$ is a winning-tree for $P_0$ with respect to the strategy $\tau_o$.
The problem of existence of winning-strategy for $P_0$ is then reduced to non-emptiness checking of $\mathcal{M}$.

Using the above, we get the following result: Given an incomplete-information quantitative game $\G$ and $\omega$-regular comparator  $\A_f$ for  the aggregate function $f$, the time complexity of determining whether $P_0$ has a winning strategy is exponential in ${|\G| \cdot |D_f|} $, where $|D_f| = |\A_f|^{O(|\A_f|)}$.

Observe that since $D_f$ is obtained by determinization of $A_f$, we obtain that $|D_f|= |\A_f|^{O(|\A_f|)}$. The thinning operation is linear in size of $|\G\times D_f|$, therefore $|\mathcal{M}| = |\G|\cdot| D_f|$. Non-emptiness checking of alternating Parity tree automata is exponential. Therefore, our procedure is doubly exponential in size of the comparator and exponential in size of the game. The question of tighter bounds is open.

\section{Discounted-sum comparator}%
\label{Sec:DiscountedSum}

The discounted-sum of an infinite sequence $A$ with discount-factor $d >1$, denoted by $\DSum{A}{d}$, is defined as $\Sigma_{i=0}^{\infty} A[i] / d^{i}$, and the discounted-sum of a finite sequence $A$ is $\Sigma_{i=0}^{|A|-1} A[i] / d^{i}$.
The discounted-sum comparator (DS-comparator, in short) for discount-factor $d$ and relation R, denoted by $ \A_{\DSsucceqFord{d}}^R$, accepts a pair $ (A, B) $ of (infinite length) weight sequences  iff $ \DSum{A}{d}$ $R$ $\DSum{B}{d} $.
We investigate properties of the DS-comparator,
and show that the DS-comparator is $\omega$-regular iff the discount-factor $d>1$ is an integer.
We also show that the discounted-sum aggregate function is $\omega$-regular iff the discount-factor is an integer.
Finally, we show the repercussions of the above results on quantitative inclusion with discounted-sum aggregate function (DS-inclusion, in short).
Section~\ref{Sec:RationalDF} and Section~\ref{Sec:IntegerDF} deal with the non-integer rational discount-factors and integer discount-factors, respectively.

\subsection{Non-integer, rational discount-factor}%
\label{Sec:RationalDF}

We prove that for non-integer discount factors, the discounted-sum comparator is not $\omega$-regular.
For a weighted $\omega$-automaton $\A$ and a real number $r \in \Re$, the \emph{cut-point language} of $\A$ w.r.t.  $r $ is defined as  $L^{\geq r} = \{w \in L(\A) | wt_\A(w) \geq r \}$~\cite{chatterjee2009expressiveness}.
When the discount factor is a rational value $1<d<2$, it is known that not all deterministic weighted $\omega$-automaton with discounted-sum aggregate function (DS-automaton, in short) have an $\omega$-regular cut-point language for an $r \in \Re$~\cite{chatterjee2009expressiveness}.
In this section, we extended this result to  all non-integer, rational discount-factors $d>1$.
Finally, we use this to prove that  discounted-sum is not an $\omega$-regular aggregate function when its discount-factor is a non-integer rational number.

\paragraph{\textbf{Ambiguous Words}}

Let $d>2$ be a non-integer, rational discount-factor. We consider finite weight-sequences over the alphabet $\{0,1,\dots, \ceil{d}-1\}$. We say a weight-sequence $w$ is \emph{ambiguous} if $1 - \frac{\ceil{d}-1}{d^{|w|-1}\cdot (d-1)} \leq \DSum{w}{d} < 1$. Intuitively, a weight-sequence is ambiguous if it could be extended to an infinite word with length less than $1$ and greater than $1$.

We will establish that there exists an infinite word such that its discounted-sum is equal to 1 but all of its finite prefixes are ambiguous. We prove by induction on length of prefixes. Let $w = 0$. Since $d>2$, $\frac{\ceil{d}-1}{d-1}>1$.  So, $w$ is ambiguous. Now, we prove that if $w$ is ambiguous, then at least one of $w\cdot 0, \dots, w\cdot (\ceil{d}-1)$ is ambiguous.

We prove this in cases:
Suppose $\DSum{w}{d} < 1-\frac{\ceil{d}-1}{d^{|w|}}$, then we show that $w\cdot (\ceil{d}-1)$ is ambiguous. First, since $w$ is ambiguous, we know that $\DSum{w}{d} \geq 1 - \frac{\ceil{d}-1}{d^{|w|-1}\cdot (d-1)}$. It is easy to show that $\DSum{w\cdot (\ceil{d}-1)}{d}  = \DSum{w}{d} + \frac{\ceil{d}-1}{d^{|w|}} \geq 1 - \frac{\ceil{d}-1}{d^{|w|-1}\cdot (d-1)} +  \frac{\ceil{d}-1}{d^{|w|}}  = 1-\frac{\ceil{d}-1}{d^{|w|}\cdot(d-1)}$. Next, since $\DSum{w}{d} < 1-\frac{\ceil{d}-1}{d^{|w|}}$ we get that $\DSum{w\cdot(\ceil{d}-1)}{d} = \DSum{w}{d} + \frac{\ceil{d}-1}{d^{|w|}} \leq 1$. Thus, $w\cdot(\ceil{d}-1)$ is ambiguous.

In the next case, suppose $1-\frac{\ceil{d}-1}{d^{|w|}} \leq \DSum{w}{d} < 1-\frac{\ceil{d}-2}{d^{|w|}}$. We will show that $w\cdot (\ceil{d}-2)$ is ambiguous. First, as show earlier, it is easy to see that $\DSum{w\cdot (\ceil{d}-2)}{d} < 1$ since $\DSum{w}{d} < 1-\frac{\ceil{d}-2}{d^{|w|}}$ has been assumed. Next, since $\DSum{w}{d} > 1-\frac{\ceil{d}-1}{d^{|w|}} $, we get that $\DSum{w\cdot (\ceil{d}-2)}{d} \geq 1-\frac{\ceil{d}-1}{d^{|w|}} + \frac{\ceil{d}-2}{d^{|w|}} = 1 - \frac{1}{d^{|w|}}$. Since $d>2$,  we know that $\frac{\ceil{d}-1}{d-1} > 1$, so $\DSum{w\cdot (\ceil{d}-2)}{d} \geq 1 - \frac{\ceil{d}-1}{d^{|w|}\cdot (d-1)}$.
Thus, $w\cdot(\ceil{d}-2)$ is ambiguous.

Similarly, when $1-\frac{i+1}{d^{|w|}} \leq \DSum{w\cdot i}{d} < 1 - \frac{i}{d^{|w|}}$, we can show that $w\cdot i$ is ambiguous for $i \in \{ 1, \ceil{d}-2\}$.

In the final case, $\DSum{w}{d} \geq 1 -\frac{1}{d^{|w|}}$. Since $\DSum{w\cdot 0}{d} = \DSum{w}{d}$, clearly, $\DSum{w\cdot 0}{d} < 1$ since $w$ is ambiguous. Further, $\DSum{w}{d} \geq 1 -\frac{1}{d^{|w|}}$, we get that $\DSum{w\cdot 0}{d} \geq 1 -\frac{\ceil{d}-1}{d^{|w|}\cdot (d-1)}$ as $\frac{\ceil{d} -1}{d-1} > 1$. Thus, $w\cdot 0$ is ambiguous.

\begin{thm}%
\label{lem:cutpoint}
For all non-integer, rational discount-factor $d>1$, there exists a deterministic discounted-sum automata $\A$ and rational value $r \in \Re$ for which its cut-point language is not $\omega$-regular.
\end{thm}
\begin{proof}

Since the proof for $1<d<2$ has been presented in~\cite{chatterjee2009expressiveness}, we skip that case.

The proof presented here extends the earlier result on $1<d<2$ from~\cite{chatterjee2009expressiveness} to all non-integer, rational discount factors $d>2$.

Let $d>2$ be a non-integer, rational discount-factor. Define deterministic discounted-sum automata $\A$ over the alphabet $\{0,1,\dots, \ceil{d}-1\}$ such that the weight of transitions on alphabet $n \in \{0,1,\dots, \ceil{d}-1\}$ is $n$. Therefore, weight of word $w \in \A$ is $\DSum{w}{d}$.

Consider its cut-point language $L^{\geq 1}$.
Let us assume that the language $L^{\geq 1}$ is $\omega$-regular and represented by B\"uchi automaton $\mathcal{B}$. For $n<m$, let the $n$- and $m$-length prefixes of $w^{\geq}$, denoted $w^{\geq}[0,n-1]$ and $w^{\geq}[0,m-1]$, respectively, be such that they reach the same states in $\mathcal{B}$. Then there exists an infinite length word $w_s$ such that $\DSum{w^{\geq}[0,n-1]\cdot w_s}{d} = \DSum{w^{\geq}[0,m-1]\cdot w_s}{d} = 1$. Now, $\DSum{w^{\geq}[0,n-1]\cdot w_s}{d} = \DSum{w^{\geq}[0,n-1]}{d} + \frac{1}{d^n}\cdot\DSum{w_s}{d}$ and $\DSum{w^{\geq}[0,m-1]\cdot w_s}{d} = \DSum{w^{\geq}[0,m-1]}{d} + \frac{1}{d^m}\cdot\DSum{w_s}{d}$. Eliminating $\DSum{w_s}{d}$ from the equations and simplification, we get:
\[
d^{m-1}\cdot(\DSum{w^{\geq}[0,m-1]}{d} - 1) + d^{n-1}\cdot(\DSum{w^{\geq}[0,n-1]}{d} - 1)  = 0
\]
The above is a polynomial over $d$ with degree $m-1$ and integer co-efficients. Specifically, $d = \frac{p}{q} > 2$ such that integers $p,q>1$, and $p$ and $q$ are mutually prime. Since $d =\frac{p}{q}$ is a root of the above equation, $q$ must divide co-efficient of the highest degree term, in this case it is $m-1$. The co-efficient of the highest degree term in the polynomial above is $(w^{\geq}[0] - (\ceil{d}-1))$.  Recall from construction of the infinite-length word with ambiguous prefixes $w^{\geq}$ from above, $w^{\geq}[0] = 0$. So the co-efficient of the highest degree term is $-1$, which is not divisible by integer $q>1$.  Contradiction.
\end{proof}

Finally, we use Theorem~\ref{lem:cutpoint} to prove the discounted-sum comparator is not $\omega$-regular when the discount-factor $d>1$ is non-integer, rational number.
\begin{thm}%
\label{Thrm:DSNotRegular}
	DS-comparator  for non-integer, rational discount-factors $d>1$ for all inequalities and equality are not $\omega$-regular.
\end{thm}
\begin{proof}
	If the comparator for an aggregate function for any one inequality is not $\omega$-regular, then the comparator for all inequalities and equality relation will also not be $\omega$-regular. Therefore, it is sufficient to prove that the discounted-sum comparator with non-integer, rational value for relation $\geq$ is not $\omega$-regular.

	Let $d>1$ be a non-integer, rational discount-fact. Let $ \A$ be the discounted-sum automaton as described in proof of Lemma~\ref{lem:cutpoint}. Consider its cut-point language $L^{\geq 1}$. From Lemma~\ref{lem:cutpoint} and~\cite{chatterjee2009expressiveness}, we know that $L^{\geq 1}$ is not an $\omega$-regular language.

	Suppose there exists an $\omega$-regular DS-comparator $\A_d^\leq$ for non-integer rational discount factor $d>1$ for relation $\geq$.
	We define the B\"uchi automaton $\mathcal{P}$ s.t. $\L(\mathcal{P}) = \{(w,v) | w \in \L(\A), v = \ceil{d}-1 \cdot 0^{\omega} \}$. Note that $\DSum{\ceil{d}-1\cdot 0^{\omega}}{d} = \ceil{d}-1$.  Then the cut-point language $L^{\geq (\ceil{d}-1)}$ of deterministic discounted-sum automata $\A$ can be constructed by taking the intersection of $\mathcal{P}$ with $\A_d^\geq$. Since all actions are closed under $\omega$-regular operations, $L^{\geq 1}$ can be represented by a B\"uchi automaton. Contradiction to Theorem~\ref{lem:cutpoint}.
\end{proof}

\begin{thm}%
	\label{Cor:RationalAggregate}
	Let $d>1$ be a non-integer, rational discount-factor. The discounted-sum aggregate function with discount-factor $d$ is not $\omega$-regular.
\end{thm}
\begin{proof}
	Immediate from Lemma~\ref{lem:cutpoint} and Theorem~\ref{thrm:functionthencomparator}.
\end{proof}
Since the DS-comparator for all non-integer, rational discount-factor $d>1$ is not $\omega$-regular, the $\omega$-regular-based algorithm for quantitative inclusion described in Algorithm~\ref{Alg:DSInclusion} does not apply to DS-inclusion. In fact, the decidability of DS-inclusion with non-integer, rational discount-factors is still open. Finally, we have shown is follow-up work that comparators for approximations of discounted-sum with non-integer discount factors $1<d<2$ can be made $\omega$-regular~\cite{bansal2022synthesis}.

\subsection{Integer discount-factor}%
\label{Sec:IntegerDF}
In this section, we provide an explicit construction of an $\omega$-regular comparator for discounted-sum with integer discount-factors. We use this construction to prove that discounted-sum aggregate function with integer discount-factor is $\omega$-regular.
Finally, we use the $\omega$-regular DS-comparator in Algorithm~\ref{Alg:DSInclusion} to establish that $\mathsf{PSPACE}$-completeness of DS-inclusion with integer discount-factors.

\paragraph{Discounted-sum comparator}
Let integer $\mu>0$ be the upper-bound on sequences.
The core intuition is that bounded sequences can be converted to their value in base $d$ via a finite-state transducer. Lexicographic comparison of the converted sequences renders the desired DS-comparator. Conversion of sequences to base $d$ requires a certain amount of \emph{look-ahead} by the transducer.
Here we describe a method that directly incorporates the look-ahead with lexicographic comparison to obtain the DS-comparator for integer discount-factor $d>1$. Here we construct the discounted-sum comparator for relation $<$.

We explain the construction in detail now.
For weight  sequence $A$ and integer discount-factor $d>1$, $\DSum{A}{d}$ can be interpreted as a value in base $d$ i.e. $\DSum{A}{d} = A[0] + \frac{A[1]}{d} + \frac{A[2]}{d^2} + \dots =  (A[0].A[1]A[2]\dots)_d$~\cite{chaudhuri2013regular}. 
Unlike comparison of numbers in base $d$, the lexicographically larger sequence may not be larger in value since (i) The elements of  weight sequences may be larger in value than base $d$, and (ii) Every value has multiple infinite-sequence representations.

To overcome these challenges, we resort to arithmetic techniques in base $d$.
Note that $ \DSum{B}{d} > \DSum{A}{d} $ iff there exists a sequence $C$ such that $\DSum{B}{d} =  \DSum{A}{d} + \DSum{C}{d} $, and $\DSum{C}{d} > 0$.
Therefore, to compare the discounted-sum of $A$ and $B$, we
 obtain a sequence $ C $.
 Arithmetic in base $d$ also results in sequence $X$ of carry elements. Then:

\begin{lem}%
 \label{Lemma:discount-invariant1}
Let $A, B, C, X$ be weight sequences, $ d > 1 $ be a positive integer such that the following equations hold:
 \begin{enumerate}
 \item\label{eq:initial}  $ A[0] + C[0] + X[0] = B[0]$
 \item\label{eq:invariant} For $ i\geq 1 $, $A[i] + C[i] + X[i] = B[i] + d \cdot X[i-1]$
 \end{enumerate}
 Then $ \DSum{B}{d} = \DSum{A}{d} + \DSum{C}{d}$.
 \end{lem}
 \begin{proof}
  $ \DSum{A}{d} + \DSum{C}{d} = \Sigma_{i=0}^{\infty} A[i] \frac1{d^i} + \Sigma_{i=0}^{\infty} C[i] \frac1{d^i} = \Sigma_{i=0}^{\infty} (A[i] + C[i]) \frac1{d^i} = (B[0] - X[0]) + \Sigma_{i=1}^{\infty} (B[i] + d\cdot X[i-1] - X[i]) \frac1{d^i} =  (B[0] - X[0]) + \Sigma_{i=1}^{\infty} (B[i] + d\cdot X[i-1] - X[i]) \frac1{d^i} = \Sigma_{i=0}^{\infty} B[i] \cdot \frac1{d^i} - \Sigma_{i=0}^{\infty} X[i] +  \Sigma_{i=0}^{\infty} X[i]  =  \Sigma_{i=0}^{\infty} B[i] \cdot \frac1{d^i} = \DSum{B}{d}$
  \end{proof}
Hence to determine $\DSum{B}{d}-\DSum{A}{d}$, systematically guess sequences $ C $ and $ X $ using the equations, element-by-element beginning with the 0-th index and moving rightwards.
There are two crucial observations here: (i)
Computation of $i$-th element of $C$ and $X$ only depends on $i$-th and $(i-1)$-th elements of $A$ and $B$. Therefore guessing $C[i]$ and $X[i]$ requires \emph{finite memory} only.
(ii) Intuitively, $C$
refers to a representation of value $\DSum{B}{d} - \DSum{A}{d}$ in base $d$ and $X$ is the carry-sequence. If we can prove that $X$ and $C$ are also bounded-sequences and can be constructed from a finite-set of integers, we would be able to further proceed to construct a B\"uchi automaton for the desired comparator.

We proceed by providing an inductive construction of sequences $C$ and $ X$ that satisfy the properties in Lemma~\ref{Lemma:discount-invariant1} (Lemma~\ref{Lemma:XInvariant}), and show that these sequences are bounded when $A$ and $B$ are bounded. In particular, when $A $ and $B$ are bounded integer-sequences, then sequences $C$ and  $X$ constructed here are also bounded-integer sequences. Therefore, they can be constructed from a finite-set of integers. Proofs for sequence $C$ are in Lemma~\ref{Lemma:SimplifyResidual}-Lemma~\ref{Lemma:BoundOnC}, and proof for sequence $X$ is in Lemma~\ref{Lemma:BoundOnX}.

We begin with  introducing some notation.
Let $\DSumDiff{i} = \Sigma^i_{j=0} (B[j] - A[j])\cdot \frac{1}{d^j} $ for all index $i\geq 0$. Also, let $ \DSumDiff{\cdot} = \Sigma^{\infty}_{j=0} (B[j] - A[j])\cdot \frac{1}{d^j} = \DSum{B}{d} - \DSum{A}{d} $. Define $ \MaxC = \max\cdot\frac{d}{d-1} $.
 We define the residual function $ \Res: \mathbb{N} \cup \{0\} \mapsto \mathbb{R} $ as follows:
 \[
 \Res(i) =
 \begin{cases}
 \DSumDiff{\cdot} - \floor{\DSumDiff{\cdot}} & \text{if } i = 0 \\
 \Res(i-1) - \floor{\Res(i-1)\cdot d^i}\cdot \frac{1}{d^i } & \text{otherwise}
 \end{cases}\]
 Then we define $C[i]$ as follows:
 \[
 C[i] =
 \begin{cases}
 \floor{\DSumDiff{\cdot}} &\text{if } i = 0 \\
 \floor{\Res(i-1)\cdot d^i } &\text{otherwise}
 \end{cases}\]
 Intuitively, $ C[i] $ is computed by \emph {stripping off} the value of the $ i $-th digit in a representation of $ \DSumDiff{\cdot} $ in base $ d $.
 $ C[i] $ denotes the numerical value of the $ i $-th position of the difference between $ B $ and $ A $. The residual function denotes the numerical value of the difference remaining after assigning the value of $ C[i] $ until that $ i $.

 We define function $ \CSum(i): \N\cup\{0\} \rightarrow \mathbb{Z}$ s.t. $\CSum(i) = \Sigma_{j=0}^{i} C[j]\cdot \frac{1}{d^j} $.
 Then, we define $  X[i] $ as follows:
 \[
 X[i] = (\DSumDiff{i} - \CSum(i)) \cdot d^i
 \]
Therefore, we have defined sequences $C$ and $X$ as above. We now prove the desired properties  one-by-one.

 First, we establish  that sequences $C$, $X$ as defined here satisfy Equations~\ref{eq:initial}-\ref{eq:invariant} from Lemma~\ref{Lemma:discount-invariant1}. Therefore, ensuring that $C$ is indeed the difference between sequences $B$ and $A$, and $X$ is their carry-sequence.
  \begin{lem}%
  \label{Lemma:XInvariant}
  Let $A$ and $B$ be bounded integer sequences and $C$ and $X$ be defined as above. Then,
  \begin{enumerate}
  \item $ B[0] = A[0] + C[0] + X[0] $
  \item For $ i \geq 1 $, $ B[i] + d\cdot X[i-1] = A[i] + C[i] + X[i] $
  \end{enumerate}
  \end{lem}

  \begin{proof}
  We prove this by induction on $ i $ using definition of function $ X $.

  When $ i = 0 $, then $ X[0] = \DSumDiff{0} - \CSum(0) \implies X[0] = B[0] - A[0] - C[0] \implies B[0] = A[0] + C[0] + X[0] $.


  When $ i = 1 $, then $ X[1] =  (\DSumDiff{1} - \CSum(1))\cdot d = (B[0] + B[1]\cdot \frac1d ) - (A[0] + A[1]\cdot\frac1d) - (C[0] + C[1]\cdot\frac1d)\cdot d \implies X[1] = B[0]\cdot d + B[1] - (A[0]\cdot d + A[1]) - (C[0]\cdot d + C[1])$. From the above we obtain $X[1] = d \cdot X[0] + B[1] - A[1] - C[1] \implies B[1] + d \cdot X[0] = A[1] + C[1]+ X[1] $.


  Suppose the invariant holds true for all $ i\leq n $, we show that it is true for $ n+1 $. $X[n+1] = (\DSumDiff{n+1} - \CSum(n+1))\cdot d^{n+1} \implies X[n+1] = (\DSumDiff{n} - \CSum(n))\cdot d^{n+1} + (B[n+1] - A[n+1] - C[n+1]) \implies X[n+1] = X[n]\cdot d + B[n+1] - A[n+1] - C[n+1] \implies B[n+1] + X[n] \cdot d = A[n+1] + C[n+1] + X[n+1]$.
  \end{proof}

  Next, we establish the sequence $C$ is a bounded integer sequence, therefore it can be represented by a finite-set of integers. First of all, by definition of $C[i]$ it is clear that $C[i]$ is an integer for all $i\geq 0$.  We are left with proving boundedness of $C$. Lemma~\ref{Lemma:SimplifyResidual}-Lemma~\ref{Lemma:BoundOnC} establish boundedness of $C[i]$.

  \begin{lem}%
  \label{Lemma:SimplifyResidual}
  For all $ i \geq 0 $, $ \Res(i) = \DSumDiff{\cdot} - \CSum(i) $.
  \end{lem}
  \begin{proof}
  Proof by simple induction on the definitions of functions $ \Res $ and $ C $.

  \begin{enumerate}
  	\item When  $i = 0$,  $\Res(0) = \DSumDiff{\cdot} - \floor{\DSumDiff{\cdot}}$. By definition of $C[0]$, $\Res(0) = \DSumDiff{\cdot} - C[0] \iff \Res(0) = \DSumDiff{\cdot} - \CSum(0)  $.

  	\item Suppose the induction hypothesis is true for all $i<n$. We prove it is true when $i = n$. When $i = n$, $\Res(n) = \Res(n-1) - \floor{\Res(n-1)\cdot d^{n} } \cdot \frac{1}{d^n}$.  By definition of $C[n]$ and I.H, we get
  	$\Res(n) = (\DSumDiff{\cdot} - \CSum(n-1)) - C[n] \cdot \frac{1}{d^n}$. Therefore
  		$\Res(n) = \DSumDiff{\cdot} - \CSum(n)$.
        \qedhere
  \end{enumerate}
  \end{proof}

  \begin{lem}%
  \label{Lemma:BoundOnRes}
If $ \DSumDiff{\cdot}\geq 0 $, then for all $ i\geq  0 $, $ 0 \leq \Res(i) < \frac{1}{d^i} $.
  \end{lem}
  \begin{proof}
   Since $ \DSumDiff{\cdot}\geq 0 $, $ \Res(0) = \DSumDiff{\cdot} - \floor{\DSumDiff{\cdot}} \geq 0 $ and $ \Res(0) = \DSumDiff{\cdot} - \floor{\DSumDiff{\cdot}} < 1 $. Specifically, $ 0 \leq \Res(0) < 1 $.

   Suppose for all $ i \leq k $, $ 0\leq  \Res(i) < \frac{1}{d^i} $. We show this is true even for $ k+1 $.

   Since $ \Res(k)\geq 0 $,  $ \Res(k)\cdot d^{k+1}  \geq 0 $. Let $ \Res(k)\cdot d^{k+1}  = x+f $, for integral $ x\geq 0 $, and fractional $ 0 \leq f < 1 $. Then, from definition of $\Res$, we get $\Res(k+1) = \frac{x+f}{d^{k+1}} - \frac{x}{d^{k+1}} \implies \Res(k+1) < \frac{1}{d^{k+1}} $.

   Also, $ \Res(k+1) \geq 0 $ since $ a - \floor{a } \geq 0 $ for all positive values of $ a $ (Lemma~\ref{Lemma:SimplifyResidual}).
  \end{proof}

  \begin{lem}%
  \label{Lemma:BoundOnC}
Let $\MaxC = \mu\cdot \frac{d}{d-1}$.  When $ \DSumDiff{\cdot} \geq 0 $, for $ i = 0 $, $0 \leq C(0) \leq \MaxC $, and for $ i \geq 1 $, $0 \leq  C(i) < d $.
  \end{lem}
  \begin{proof}
  Since both $ A $ and $ B $ are non-negative bounded weight sequences, maximum value of $ \DSumDiff{\cdot} $ is when $ B = \{\max\}_{i} $ and $ A = \{0\}_i $. In this case $ \DSumDiff{\cdot} = \MaxC $. Therefore, $ 0 \leq C[0] \leq \MaxC $.

  From Lemma~\ref{Lemma:BoundOnRes}, we know that for all $ i$, $ 0 \leq \Res(i) < \frac{1}{d^i} $. Alternately, when $ i \geq 1 $, $ 0 \leq \Res(i-1) < \frac{1}{d^{i-1}} \implies 0 \leq \Res(i-1) \cdot d^i < \frac{1}{d^{i-1}} \cdot d^i \implies 0 \leq \Res(i-1)\cdot d^i < d \implies 0 \leq \floor{\Res(i-1)\cdot d^i} < d \implies 0 \leq C[i] < d$.
  \end{proof}
Therefore, we have established  that sequence $C$ is non-negative integer-valued and is bounded by $\MaxC = \mu\cdot \frac{d}{d-1}$.

 Finally,  we prove that sequence $X$ is also a bounded-integer sequence, thereby proving that it is bounded, and can be represented with a finite-set of integers. Note that for all $i\geq 0$, by expanding out the definition of $X[i]$ we get that $X[i]$ is an integer for all $i\geq 0$.
 We are left with proving boundedness of $X$:

 \begin{lem}%
 \label{Lemma:BoundOnX}
 Let $ \MaxX = 1 + \frac{\max}{d-1} $. When $ \DSumDiff{\cdot} \geq 0 $, then for all $ i\geq 0 $, $ |X(i)| \leq \MaxX $.
 \end{lem}
 \begin{proof}
 From definition of $ X $, we know that $  X(i) = (\DSumDiff{i} - \CSum(i)) \cdot d^i  \implies  X(i) \cdot \frac{1}{d^i} =  \DSumDiff{i} - \CSum(i) $. From Lemma~\ref{Lemma:SimplifyResidual} we get $ X(i) \cdot \frac{1}{d^i} =  \DSumDiff{i} - (\DSumDiff{\cdot} - \Res(i)) \implies X(i) \cdot \frac{1}{d^i} =  \Res(i) - (\DSumDiff{\cdot} - \DSumDiff{i}) \implies  X(i) \cdot \frac{1}{d^i} =  \Res(i) -
  (\Sigma_{j=i+1}^{\infty}(B[j]-A[j])\cdot \frac{1}{d^j}) \implies  |X(i) \cdot \frac{1}{d^i}| \leq |\Res(i)| + |(\Sigma_{j=i+1}^{\infty}(B[j]-A[j])\cdot \frac{1}{d^j}) | \implies  |X(i) \cdot \frac{1}{d^i}| \leq |\Res(i)| + \frac{1}{d^{i+1}}\cdot|(\Sigma_{j=0}^{\infty}(B[j+i+1]-A[j+i+1])\cdot \frac{1}{d^j}) | \implies  |X(i) \cdot \frac{1}{d^i}| \leq |\Res(i)| + \frac{1}{d^{i+1}}\cdot|\MaxC|$. From Lemma~\ref{Lemma:BoundOnRes}, this implies  $|X(i) \cdot \frac{1}{d^i}| \leq \frac{1}{d^i} + \frac{1}{d^{i+1}}\cdot|\MaxC| \implies  |X(i)| \leq 1 + \frac{1}{d}\cdot|\MaxC| \implies  |X(i)| \leq 1 + \frac{\max}{d-1} \implies |X(i)| \leq \MaxX$
 \end{proof}

 We summarize our results from Lemma~\ref{Lemma:XInvariant}-Lemma~\ref{Lemma:BoundOnX} as follows:
 \begin{cor}%
 	\label{lem:FiniteAlphabetXandC}
 	Let $d>1$ be an integer discount-factor. Let $A$ and $B$ be non-negative
 	integer sequences bounded by $\mu$, and $\DSum{A}{d} <
 	\DSum{B}{d}$. Then there exists bounded integer-valued sequences $X$ and $C$ that satisfy the conditions in Lemma~\ref{Lemma:discount-invariant1}. Furthermore, $C$ and $X$ are bounded as follows:
 	\begin{enumerate}
 		\item\label{Item1:BoundOnC}  $0 \leq C[0] \leq \mu\cdot \frac{d}{d-1}$ and for all $i\geq 1$, $0\leq C[i] < d$,
 		\item For all $i \geq 0$, $0 \leq |X[i]| \leq 1 + \frac{\mu}{d-1}$
 	\end{enumerate}
 \end{cor}

\noindent
Intuitively, we construct a B\"uchi automaton $\A_{\DSsucceqFord{d}}^< $ with states of the form $(x,c)$ where $x$ and $c$ range over all possible values of $X$ and $C$, respectively, and a special initial state $s$. Transitions over alphabet $(a,b)$ replicate the equations in Lemma~\ref{Lemma:discount-invariant1}.
i.e.\ transitions from the start state $(s,(a,b), (x,c))$ satisfy  $a+c+x = b$ to replicate Equation~\ref{eq:initial} (Lemma~\ref{Lemma:discount-invariant1}) at the 0-th index, and all other transitions $((x_1, c_1), (a,b), (x_2, c_2))$ satisfy $a+ c_2 + x_2 = b + d\cdot x_1$ to replicate Equation~\ref{eq:invariant} (Lemma~\ref{Lemma:discount-invariant1}) at indexes $i>0$.
Full construction is as follows:

\renewcommand{\max}{\mu}
\renewcommand{\MaxC}{\mu_C}
\renewcommand{\MaxX}{\mu_X}

\paragraph{Construction}
Let $ \MaxC = \max \cdot \frac{d}{d-1}  $ and $ \MaxX = 1 + \frac{\max}{d-1}$.
 $ \A_{\DSsucceqFord{d}}^< = (\Statess, \Sigma, \delta_d, \StartState, \Final) $
\begin{itemize}

\item $ \Statess =  \{s\} \cup \AcceptingStates \cup S_{\bot} $ where \\
$\Final = \{(x,c) ||x| \leq \MaxX, 0 \leq c \leq \MaxC \} $, and \\
$S_{\bot} = \{(x, \bot) | | x| \leq \MaxX\}$ where $\bot$ is a special character, and $ c \in \N$, $x \in \mathbb{Z}$.


\item State $s$ is the initial state, and  $\Final$ are accepting states

\item $ \Sigma = \{(a,b) : 0 \leq a, b \leq \max \} $ where $ a $ and $ b $ are integers.

\item $\delta_d \subseteq \Statess \times \Sigma \times \Statess$ is defined as follows:
	\begin{enumerate}
	\item Transitions from start state $ s $:
	\begin{enumerate}[label = \roman*]
	\item $ (s ,(a,b), (x,c)) $ for all $ (x,c) \in \AcceptingStates $ s.t. $ a + x + c = b $ and $ c \neq 0 $

	\item $ (s ,(a,b), (x, \bot)) $ for all $ (x, \bot) \in S_{\bot} $ s.t. $ a + x  = b $

	\end{enumerate}

	\item Transitions within $ S_{\bot} $: $ ((x, \bot) ,(a,b), (x', \bot) )$ for all $(x, \bot)$, $(x', \bot) \in S_{\bot} $, if  $ a + x' = b + d\cdot x $

	\item Transitions within $ \Final $: $ ((x,c) ,(a,b), (x',c') )$ for all $ (x,c)$, $(x',c') \in \Final $ where $ c' < d $, if $ a + x' + c' = b + d\cdot x $

	\item Transition between $ S_{\bot} $ and $ \Final $: $ ((x, \bot),(a,b), (x',c')) $ for all $ (x,\bot) \in S_{\bot} $, $ (x',c') \in \Final $ where $ 0 < c' < d $, if  $ a + x' + c' = b + d\cdot x $

	\end{enumerate}

\end{itemize}

\begin{thm}%
	\label{thm:Construction}
	Let $d>1$ be an integer discount-factor, and $\mu>1$ be an integer upper-bound.
	B\"uchi automaton $ \A_{\DSsucceqFord{d}}^<$ accepts pair of bounded sequences $(A,B)$ iff $\DSum{A}{d}\leq \DSum{B}{d}$.  The B\"uchi automaton has $\mathcal{O}(\frac{\mu^2}{d})$-many states.
\end{thm}
\begin{proof}
Corollary~\ref{lem:FiniteAlphabetXandC} proves that if $\DSum{A}{d}< \DSum{B}{d}$ then sequences $X$ and $C$ satisfying the integer sequence criteria and bounded-criteria will exist. Let these sequences be $X = X[0]X[1]\dots$ and $C = [0]C[1]\dots$.
Since $\DSum{C}{d} > 0$, there exists an index $i\geq 0$ where $C[i]>0$. Let the first position where $C[i]>0$ be index $j$.
By construction of $\A_{\DSsucceqFord{d}}^<$, the state sequence given by $s,(X[0],\bot)\dots, (X[j-1],\bot),(X[j], C[j]),(X[j+1], C[j+1])\dots $, where for all $i\geq j$, $C[i]\neq \bot$, forms a run of word $(A,B)$ in the B\"uchi automaton. Furthermore, this run is accepting since state $(x,c)$ where $c \neq \bot$ are accepting states. Therefore, ($A,B) $ is an accepting word in $\A_{\DSsucceqFord{d}}^<$.

 To prove the other direction, suppose the pair of sequences $(A, B)$ has an accepting run with state sequence $s$, $(x_0,\bot),\dots (x_{j-1}, \bot), (x_j,c_j), (x_{j+1}, c_{j+1})\dots $, where for all $i\geq j$, $c_j \neq \bot$. Construct sequences X and C as follows:
 For all $i\geq 0$, $X[i] = x_i$. For all $i<j$, $C[i] = 0$ and for all $i\geq j$ $C[i]=c_i$.
 Then the transitions of  $\A_{\DSsucceqFord{d}}^<$ guarantees  Equations~\ref{eq:initial}-~\ref{eq:invariant}
 from   Lemma~\ref{Lemma:discount-invariant1} to hold for sequences $A$,$B$ and $C$,$X$. Therefore, it must be the case that  $\DSum{B}{d} = \DSum{A}{d} + \DSum{C}{d}$.
 Furthermore, since the first transition to accepting states $(x,c)$ where $c \neq \bot$ is possible only if $c>0$, $\DSum{C}{d}>0$. Therefore, $ \DSum{A}{d}<\DSum{B}{d}$.
 Therefore,  $\A_{\DSsucceqFord{d}}^< $ accepts $(A,B)$ if $ \DSum{A}{d}<\DSum{B}{d}$.
\end{proof}

\begin{cor}%
	\label{thm:discountedSumComparator}
	DS-comparator  for integer discount-factors $d>1$ for all inequalities and equality are  $\omega$-regular.
\end{cor}
\begin{proof}
	Immediate from Theorem~\ref{thm:Construction}, and closure properties of B\"uchi automaton.
\end{proof}
Constructions of DS-comparator with integer discount-factor $d>1$ for non-strict inequality $\leq$ and equality $=$ follow similarly and also have $\mathcal{O}(\frac{\mu^2}{d})$-many states.

\paragraph{Discounted-sum aggregate function}
We use the $\omega$-regular comparator for DS-aggregate function for integer discount-factor to prove that discounted-sum with integer discount-factors is an $\omega$-regular aggregate function.
\begin{thm}%
	\label{thm:DSRegular}
	Let $d>1$ be an integer discount-factor. The discounted-sum aggregate function with discount-factor $d$ is $\omega$-regular under base $d$.
\end{thm}
 \begin{proof}
 	We define the discounted-sum aggregate function automaton (DS-function automaton, in short): For integer $\mu>0$, let $\Sigma = \{0,1,\dots \mu\}$ be the input alphabet of DS-function, and  $d>1$ be its integer base. B\"uchi automaton $\A^\mu_d$ over alphabet $\Sigma\times\mathsf{AlphaRep}(d)$ is a DS-function automaton of type $\Sigma^\omega \rightarrow \Re$ if for all $A \in \Sigma^\omega$, $ (A, \mathsf{rep}(\DSum{A}{d}), d) \in \A^\mu_d$. Here we prove that such a $\A_d^\mu$ exists.

	Let $\mu>0$ be the integer upper-bound. Let $\A^=_d$ be the DS-comparator for integer discount-factor $d>1$ for relation $=$. Intersect  $\A^=_d $ with the B\"uchi automata consisting of all infinite words from alphabet $\{0,1\dots \mu\}\times \{0,\dots,d-1\}$. The resulting automaton $\mathcal{B}$ accepts $(A,B)$ for $A \in \{0,\dots,\mu\}^\omega$ and $B\in  \{0,\dots,d-1\}^\omega$ iff $\DSum{A}{d} = \DSum{B}{d}$.  Next, we want to ensure that for all $(A,B)$ accepted by $\mathcal{B}$, $B$ is not of the form $\{0,\dots, d-1\}^*\cdot (d-1)^\omega$. I.e., either $\DSum{B}{d}$ is an irrational number or $\DSum{B}{d}$ is a rational number in which case we prevent the finite-representation which ends in an infinite series of $(d-1)$. This can be represented by the B\"uchi automata $\mathcal{B} \setminus \mathcal{C}$, where $\mathcal{C}$ accepts words of the form $(A,B)$ s.t. $A \in \{0,\dots, \mu\}^\omega$ and $B \in \{0, \dots, d-1\}^* \cdot (d-1)^\omega$. A non-deterministic B\"uchi automaton for $\mathcal{C}$ can be constructed by hand.  The automaton for $\mathcal{B}\setminus\mathcal{C}$ will be a Parity automaton.

	Since all elements of $B$ are bounded by $d-1$, $\DSum{B}{d}$ can be represented as an $\omega$-word as follows: Let $B = B[0],B[1]\dots$, then  its  $\omega$-word representation in base $d$  is given by $ +\cdot (\mathsf{Int}(\DSum{B}{d}, d), \mathsf{Frac}(\DSum{B}{d}, d))$ where $\mathsf{Int}(\DSum{B}{d}, d) = B[0] \cdot 0^\omega$ and $ \mathsf{Frac}(\DSum{B}{d}, d) = B[1],B[2]\dots$. This transformation of integer sequence $B$ into its $\omega$-regular word form in base $d$ can be achieved with a simple transducer $\mathcal{T}$.

	Therefore, application of transducer $\mathcal{T}$ to Parity automaton $\mathcal{B}\setminus \mathcal{C}$ will result in a Parity automaton over the alphabet $\Sigma\times\mathsf{AlphaRep}(d)$ such that for all $A \in \Sigma^\omega$ the automaton accepts $(A, \mathsf{rep}(\DSum{A}{d},d))$. This is exactly the DS-function automaton over input alphabet $\Sigma$ and integer base $d>1$. Therefore, the discounted-sum aggregate function with integer discount-factors in $\omega$-regular.
 \end{proof}
 Recall, this proof works only for the discounted-sum aggregate function with integer discount-factor. In general, there is no known procedure to derive a function automaton from an $\omega$-regular comparator (Conjecture~\ref{Conjecture:comparatortofunction}).

 \paragraph{DS-inclusion}
 For discounted-sum with integer discount-factor it is in \textsf{EXPTIME}~\cite{boker2014exact,chatterjee2010quantitative} which does not match with its existing  \textsf{PSPACE} lower bound. In this section, we use the $\omega$-regular DS-comparator for integer to close the gap, and establish \cct{PSPACE-completeness} of DS-inclusion under a unary representation of numbers.
 \begin{cor}%
 	\label{Cor:DSInSEq}
 	Let integer $\mu>1$ be the maximum weight on transitions in DS-automata $P$ and $Q$, and $d>1$ be an integer discount-factor. Let $\mu$ and $d$ be represented in unary form. Then DS-inclusion, DS-strict-inclusion, and DS-equivalence between $P$ and $Q$ are $\cc{PSPACE}$-complete.

 \end{cor}
 \begin{proof}
  Since size of DS-comparator is polynomial w.r.t.\ to upper bound $\mu$, when represented in unary, (Theorem~\ref{thm:Construction}),  DS-inclusion is $\cc{PSPACE}$ in size of input weighted $\omega$-automata and $\mu$ (Theorem~\ref{thrm:RegularComplexity}).
  \end{proof}

  Not only does this result improve upon the previously known upper bound of $\cc{EXPTIME}$ but it also closes the gap between upper and lower bounds for DS-inclusion. Note, however, if the numbers are represented in binary, then the comparator-based algorithm will incur prohibitively large overhead since the size of the comparator will be exponential in $\mu$.

 We observe the algorithmic benefits of comparator-based solutions.
 The earlier known \cct{EXPTIME} upper bound in complexity is based on an exponential determinization construction (subset construction) combined with arithmetical reasoning~\cite{boker2014exact,chatterjee2010quantitative}. We observe that the determinization construction can be performed on-the-fly in $\cc{PSPACE}$. To perform, however, the arithmetical reasoning on-the-fly in \cct{PSPACE} would require essentially using the same bit-level ($(x,c)$-state) techniques that we have used to construct DS-comparator. This point is corroborated in empirical evaluations where comparator-based approach comprehensively outperforms the determinization-based approach~\cite{bansal2018automata}. The performance of comparator-based approach has further been improved using additional language-theoretic properties of DS comparators, namely their safety and co-safety properties~\cite{bansal2019safety}.

\section{Limit-average comparator}%
\label{Sec:LA}
The limit-average of an infinite sequence $M$ is the point of convergence of the average of prefixes of $M$.
Let $\Sum{M[0,n-1]}$ denote the sum of the $n$-length prefix of sequence $M$.
The \emph{limit-average infimum}, denoted by $\LAInf{M}$, is defined as $\liminf{n}{\infty}{\frac{1}{n}\cdot \Sum{M[0,n-1]}}$. Similarly, the \emph{limit-average supremum}, denoted by $\LASup{M}$, is defined as $\limsup{n}{\infty}{\frac{1}{n}\cdot \Sum{M[0,n-1]}}$.
The limit-average of sequence $M$, denoted by $\LA{M}$, is defined \emph{only if} the limit-average infimum and limit-average supremum coincide, and then $\LA{M} = \LAInf{M}$ ($=\LASup{M}$).  Note that while limit-average infimum and supremum exist for all bounded sequences, the limit-average may not.
To work around this limitation of limit-average, most applications simply use limit-average infimum or limit-average supremum of sequences~\cite{brim2011faster,chatterjee2010quantitative,chatterjee2005mean,zwick1996complexity}.
However, the usage of limit-average infimum or limit-average supremum in lieu of limit-average for purpose of comparison can be misleading.
For example, consider sequence $A$ s.t. $\LASup{A} = 2$ and $\LAInf{A} = 0$, and sequence $B$ s.t. $\LA{B} = 1$. Clearly,  limit-average of $A$ does not exist. So while it is true that $ \LAInf{A} < \LAInf{B}$, indicating that at infinitely many indices the average of prefixes of $A$ is lower, this renders an incomplete picture since at infinitley many indices, the average of prefixes of $B$ is greater as $\LASup{A} = 2 $.

Such inaccuracies in limit-average comparison may occur when the limit-average of at least one sequence does not exist. However, it is not easy to distinguish sequences for which limit-average exists from those for which it doesn't.

We define \emph{prefix-average comparison} as a relaxation of limit-average comparison.
Prefix-average comparison coincides with limit-average comparison when limit-average exists for both sequences.
Otherwise, it determines whether eventually the average of prefixes of one sequence are greater than those of the other.
This comparison does not require the limit-average to exist to return intuitive results.
Further, we show that the \emph{prefix-average comparator} is $\omega$-context-free.

\subsection{Limit-average language and comparison}%
\label{Sec:LALanguage}

Let $\Sigma = \{0,1,\dots, \mu\}$ be a finite alphabet with $\mu>0$. The \emph{limit-average language} $\L_{LA}$ contains the sequence (word) $A \in \Sigma^{\omega}$ iff its limit-average exists. We begin with the intuition to why limit-average language is neither $\omega$-regular nor $\omega$-context free. The formal argument is available in Theorem~\ref{Lemma:LARegularNotExist}.

Suppose $\L_{LA}$ were $\omega$-regular, then $\L_{LA} = \bigcup_{i=0}^n U_i\cdot V_i^{\omega}$, where $U_i, V_i\subseteq \Sigma^*$ are regular languages over \emph{finite} words.
The limit-average of sequences is determined by its behavior in the limit, so limit-average of sequences in $V_i^{\omega}$  exists.
Additionally, the average of all (finite) words in $V_i$ must be the same.
If this were not the case, then two words in $V_i$ with unequal averages $l_1$ and $l_2$, can generate a word $w \in V_i^{\omega}$ s.t~the average of its prefixes oscillates between $l_1$ and $l_2$. This cannot occur, since limit-average of $w$ exists.
Let the average of sequences in $V_i$ be $a_i$, then limit-average of  sequences in $V_i^{\omega}$ and  $U_i \cdot V_i^{\omega}$ is also $a_i$.
This is contradictory since
there are sequences with limit-average different from   $a_i$.
Similarly, since every $\omega$-CFL is  represented by $\bigcup_{i=1}^n U_i \cdot V_i^{\omega}$ for CFLs $U_i, V_i$ over finite words~\cite{cohen1977theory}, a similar argument proves that $\L_{LA}$ is not $\omega$-context-free.

\begin{thm}%
\label{Lemma:LARegularNotExist}
$\L_{LA}$ is neither an $\omega$-regular nor an $\omega$-context-free language.
\end{thm}
\begin{proof}

We first prove that {$\L_{LA}$ is not $\omega$-regular}.

	Let us assume that the language $\L_{LA}$ is $\omega$-regular. Then there exists a finite number $n$ s.t. $\L_{LA} = \bigcup_{i=0}^n U_i \cdot V_i^{\omega}$, where $U_i$ and $V_i \in \Sigma^{*}$ are regular languages over finite words.

	 For all $i \in \{0,1,\dots n\}$, the limit-average of any word in $U_i\cdot V_i^{\omega}$ is given by the suffix of the word in $V_i^{\omega}$. Since $U_i\cdot V_i^{\omega} \subseteq \L_{LA}$, limit-average exists for all words in $U_i \cdot V_i^{\omega}$. Therefore, limit-average of all words in $V_i^{\omega}$ must exist. As discussed above, we conclude that the average of all words in $V_i$ must be the same. Furthermore, we know that the limit-average of all words in $V_i^{\omega}$ must be the same, say $\LA{w} = a_i$ for all $w \in V_i^{\omega}$.

	 Then the limit-average of all words in $\L_{LA}$ is one of $a_0, a_1\dots a_n$. Let $a = \frac{p}{q}$ s.t $p<q$, and $a \neq a_i$ for $i \in \{0,1,\dots,\mu\}$. Consider the word $w = (1^{p}0^{q-p})^{\omega}$. It is easy to see that $\LA{w} = a$. However, this word is not present in $\L_{LA}$ since the limit-average of all words in $\L_{LA}$ is equal to $a_0$ or $a_1$ \dots or $a_n$.

	 Therefore, our assumption that $\L_{LA}$ is an $\omega$-regular language has been contradicted.

Next we prove that $\L_{LA}$ is not an $\omega$-CFL\@.

Every $\omega$-context-free language can be written in the form of $\bigcup_{i=0}^n U_i \cdot V_i^{\omega}$ where $U_i$ and $V_i $ are context-free languages over finite words.
	The rest of this proof is similar to the proof for non-$\omega$-regularity of $\L_{LA}$.
\end{proof}

In the next section, we will define \emph{prefix-average comparison} as a relaxation of limit-average comparison. To show how prefix-average comparison relates to limit-average comparison, we will require the following two lemmas:
Quantifiers  $\exists^{\infty}i$ and $\exists^{f}i$ denote the existence of \emph{infinitely} many and \emph{only finitely} many indices $i$, respectively.

\begin{lem}%
\label{Lemma:LAExistsProp}
Let $A$, $B$ be sequences s.t their limit-average exists.
If $\LA{A} >  \LA{B}$ then $\ExistFin{i}{B}{A}$.

\end{lem}
\begin{proof}

	Let the limit-average of sequences $A$, $B$ be $L_a$, $L_b$ respectively. Since the limit average of both $A$ and $B$ exists, for every $\epsilon >0$, there exists $N_{\epsilon}$ s.t.\ for all $n> N_{\epsilon}$, $|\Av{A[1,n]} - L_a| < \epsilon$ and $|\Av{B[1,n]} - L_b| < \epsilon$.

    Let $ L_a - L_b = k>0$. Let us take $\epsilon = \frac{k}{4}$. Then, for all $n > N_{\frac{k}{4}}$, we get that  	$\Av{A[0,n-1]} - \Av{B[0,n-1]} > \frac{k}{2}$, since $L_a - L_b = k$ and $|\Av{A[1,n]} - L_a| < \frac{k}{4}$ and $|\Av{B[1,n]} - L_b| < \frac{k}{4}$. Thus, for all $n > N_{\frac{k}{4}}$, we get that $\Sum{A[0,n-1]} > \Sum{B[0,n-1]}$.
    In particular, we get that $\ExistFin{i}{B}{A}$.
\end{proof}
The implication does not hold the other way since for sequences $A$ and $B$ with equal limit-average it is possible that $\exists^{\infty} i, \Sum{A[0,n-1]} > \Sum{B[0,n-1]}$ and $\exists^{\infty} i, \Sum{B[0,n-1]} > \Sum{A[0,n-1]}$.

\begin{lem}%
	\label{Lemma:LAExistFull}
Let $A$ and $B$ be sequences s.t.~their limit average exists.
If $\ExistFin{i}{B}{A}$, then $\LA{A} \geq \LA{B}$.

\end{lem}
\begin{proof}

We prove by contradiction.

Suppose, $\LA{A} < \LA{B}$. Then, from Lemma~\ref{Lemma:LAExistsProp}, we know that \[\ExistFin{i}{A}{B}.\]
But, $\ExistFin{i}{A}{B}$ and $\ExistFin{i}{B}{A}$ cannot hold together since the sequences $A$ and $B$ are of infinite length. Hence, contradiction.
\end{proof}

\subsection{Prefix-average comparison and comparator}%
\label{Sec:LAClassificationAll}

The previous section relates limit-average comparison with the sums of equal length prefixes of the sequences (Lemma~\ref{Lemma:LAExistsProp}-\ref{Lemma:LAExistFull}). The comparison criterion is based on the number of times sum of prefix of one sequence is greater than the other, which does not rely on the existence of limit-average.
Unfortunately, this comparison criterion is not necessary and sufficient for limit-average comparison due to the one-way implication of Lemma~\ref{Lemma:LAExistsProp}.
Instead, we use this criteria to define \emph{prefix-average comparison}. In this section, we define prefix-average  comparison and explain how it relaxes limit-average comparison. Lastly, we construct the prefix-average comparator, and prove that it is not $\omega$-regular but is $\omega$-context-free.

\begin{defi}[Prefix-average comparison]%
\label{Def:PrefixAverage}
Let $A$ and $B$ be number sequences. We say $\PLA{A} \succeq \PLA{B}$ if
$\ExistFin{i}{B}{A}$.
\end{defi}
Note, the definition implies that
$\ExistInf{i}{A}{B}$.

Intuitively, prefix-average comparison states that $\PLA{A}\succeq \PLA{B}$ if eventually the sum of prefixes of $A$ are always greater than those of $B$. We use $\succeq$ since the average of prefixes may be equal when the difference between the sum converges to 0.
 Definition~\ref{Def:PrefixAverage} and  Lemma~\ref{Lemma:LAExistsProp}-\ref{Lemma:LAExistFull} relate  limit-average comparison  and prefix-average comparison:

\begin{cor}
When limit-average of $A$ and $B$ exists, then%
\label{Coro:PLAtoLA}

\begin{itemize}
	\item $\PLA{A}\succeq\PLA{B} \implies \LA{A}\geq \LA{B}$.
	\item $\LA{A}>\LA{B} \implies \PLA{A}\succeq \PLA{B}$.
\end{itemize}
\end{cor}
\begin{proof}
	The first item falls directly from Definition~\ref{Def:PrefixAverage} and  Lemma~\ref{Lemma:LAExistFull}.
	The second item falls directly from
	 Definition~\ref{Def:PrefixAverage} and  Lemma~\ref{Lemma:LAExistsProp}.
\end{proof}

Therefore, on sequences for which  limit-average exists, the first bullet says that prefix-average comparison returns the same result as limit-average comparison. In addition, when limit-average may not exist, the prefix-average comparison can return intuitive results. For example,
suppose limit-average of $A$ and $B$ do not exist, but $\LAInf{A} > \LASup{B}$, then $\PLA{A} \succeq \PLA{B}$.
This way, prefix-average comparison relaxes limit-average comparison.

The rest of this section describes \emph{prefix-average comparator}, denoted by $\A_{\PASuc}^\succeq$, an automaton that accepts the pair $(A,B)$  of sequences iff $\PLA{A}\succeq \PLA{B}$.
\begin{lem}%
	\label{Lemma:PumpingLemmaRegular}
	\textbf{(Pumping Lemma for $\omega$-regular language~\cite{alur2009omega})}
	Let $L$ be an $\omega$-regular language. There exists $p \in \mathbb{N}$
	such that, for each $w = u_1w_1u_2w_2 \dots \in L$ such that $|w_i| \geq p$ for all $i$, there are
	sequences of finite words $(x_i)_{i\in \N}$, $(y_i)_{i\in \N}$, $(z_i)_{i\in \N}$ s.t., for all $i$, $w_i = x_{i}y_{i}z_i$, $|x_{i}y_i| \leq p$ and $|y_i| > 0$ and for every sequence of pumping factors $(j_i)_{i\in \N} \in \N$, the pumped word $u_1 x_1 y_1^{j_1}z_1 u_2 x_2 y_2^{j_2}z_2 \dots \in L$.
\end{lem}

\begin{thm}%
\label{Lemma:LANotRegular}
The prefix-average comparator is not $\omega$-regular.
\end{thm}
\begin{proof}[Proof Sketch]
We use Lemma~\ref{Lemma:PumpingLemmaRegular} to prove that $\A_{\PASuc}^\succeq$ is not $\omega$-regular.
Suppose $\A_{\PASuc}^\succeq$ were $\omega$-regular.
For $p >0 \in \N$, let $w = (A,B) = ((0,1)^p(1,0)^{2p})^{\omega}$.
The segment $(0,1)^*$ can be pumped s.t the resulting word is no longer in $\A_{\PASuc}^\succeq$.

Concretely, $A = (0^p 1^{2p})^{\omega}$, $B = (1^p0^{2p})^{\omega}$, $\LA{A} = \frac{2}{3}$, $\LA{B} = \frac{1}{3}$.
So, $w = (A,B) \in \A_{\PASuc}^\geq$.
Select as factor $w_i$ (from Lemma~\ref{Lemma:PumpingLemmaRegular}) the sequence $(0,1)^p$.
Pump each $y_i$ enough times so that the resulting word is $\hat{w} = (\hat{A}, \hat{B}) = ((0,1)^{m_i}(1,0)^{2p})^{\omega}$  where $m_i>4p$.
It is easy to show that $\hat{w} = (\hat{A}, \hat{B})\notin \A_{\PASuc}^\succeq$.
\end{proof}
We discuss key ideas and sketch the construction of the prefix average comparator. 
The term \emph{prefix-sum difference at $i$} indicates $\Sum{A[0,i-1]} - \Sum{B[0,i-1]}$, i.e.\ the difference between sum of $i$-length prefix of $A$ and $B$.

\paragraph{Key ideas}
For sequences $A$ and $B$ to satisfy $\PLA{A}\succeq\PLA{ B}$, $\ExistFin{i}{B}{A}$.
This implies there exists an index $N$ s.t.\ for all indices $i>N$, $\Sum{A[0,i-1]} - \Sum{B[0,i-1]}>0$.
While reading a word, the prefix-sum difference is maintained by states and the stack of $\omega$-PDA:\@ states maintain whether it is negative or positive,
while number of tokens in the stack equals its absolute value.
The automaton non-deterministically guesses the aforementioned index $N$,
beyond which the automaton ensures that prefix-sum difference remains positive.

\paragraph{Construction sketch}

The push-down comparator $\A_{\PASuc}^\succeq$ consists of three states: (i) State $s_P$ and (ii) State $s_N$ that indicate that the prefix-sum difference is greater than zero or not respectively, (iii)  accepting state $s_F$.
An execution of $(A,B)$ begins in state $s_N$ with an empty stack. On reading letter $(a,b)$, the stack pops or pushes $|(a-b)|$ tokens from the stack depending on the current state of the execution. From state $s_P$, the stack pushes tokens if $(a-b)>0$, and pops otherwise. The opposite occurs in state $s_N$.
State transition between $s_N$ and $s_P$ occurs only if the stack action is to pop but the stack consists of $k < |a-b|$ tokens. In this case, stack is emptied, state transition is performed and $|a-b|-k$ tokens are pushed into the stack.
For an execution of $(A,B)$ to be an accepting run, the automaton non-deterministically transitions into state $s_F$.
State $s_F$ acts similar to state $s_P$ except that execution is terminated if there aren't enough tokens to pop out of the stack.
$\A_{\PASuc}^\succeq$ accepts by accepting state.

To see why the construction is correct, it is sufficient to prove that at each index $i$, the number of tokens in the stack is equal to $|\Sum{A[0,i-1]} - \Sum{B[0,i-1]}|$. Furthermore, in state $s_N$, $\Sum{A[0,i-1]} - \Sum{B[0,i-1]}\leq 0$, and in state $s_P$ and $s_F$, $\Sum{A[0,i-1]} - \Sum{B[0,i-1]}>0$.
Next, the index at which the automaton transitions to the accepting state $s_F$ coincides with index $N$. The execution is accepted if it has an infinite execution in state $s_F$, which allows transitions only if $\Sum{A[0,i-1]} - \Sum{B[0,i-1]} > 0$.

\paragraph{Construction}
We provide a sketch of the construction of the B\"uchi push-down autoamaton $\A_{\PASuc}^\geq$, and then prove that it corresponds to the prefix average comparator.

Let $\mu$ be the bound on sequences. Then  $\Sigma = \{0,1,\dots, n\}$ is the alphabet of sequences.
Let $ \A_{\PASuc}^\succeq = (\Statess, \Sigma\times\Sigma, \Gamma, \delta, s_0, Z_0)$ where:

\begin{itemize}
	\item $\Statess = \{s_N, s_P, s_F\}$ is the set of states of the automaton.

	\item $\Sigma\times\Sigma$ is the alphabet  of the language.

	\item $\Gamma = \{Z_0, \alpha\}$ is the  alphabet.

	\item $s_0 = s_N$ is the start state of the  automata.

	\item $Z_0$ is the start symbol of the stack.

	\item $s_F$ is the accepting state of the automaton. Automaton $\A_{\PASuc}^\geq$ accepts words by final state.

	\item Here we give  a sketch of the behavior of the transition function $\delta$.

	\begin{itemize}
		\item When $\A_{\PASuc}^\succeq$ is in configuration $(s_P, \tau)$ for $\tau \in \Gamma$, push $a$ number of $\alpha$-s into the stack.

		Next, pop $b$ number of $\alpha$-s. If after popping $k$ $\alpha$-s where $k<b$, the PDA's configuration becomes $(s_P, Z_0)$, then first move to state $(s_N, Z_0)$ and then resume with pushing $ b-k$ $\alpha$-s into the stack.

		\item When $\A_{\PASuc}^\succeq$ is in configuration $(s_N, \tau)$ for $\tau \in \Gamma$, push $b$ number of $\alpha$-s into the stack

		Next, pop $a$ number of $\alpha$-s. If after popping $k$ $\alpha$-s where $k<a$, the PDA's configuration becomes $(s_N, Z_0)$, then first move to state $(s_P, Z_0)$ and then resume with pushing $ a-k$ $\alpha$-s into the stack.

		\item When $\A_{\PASuc}^\succeq$ is in configuration $(s_P, \tau)$ for $\tau \neq Z_0$, first move to configuration $(s_F, \tau)$ and then push $a$ number of $\alpha$-s and pop $b$ number of $\alpha$-s. Note that there are no provisions for popping $\alpha$ if the stack hits $Z_0$ along this transition.

		\item When $\A_{\PASuc}^\succeq$ is in configuration $(s_F, \tau)$ for $\tau \neq Z_0$, push $a$ $\alpha$-s then pop $b$ $\alpha$-s.

		Note that there are no provisions for popping $\alpha$ if the stack hits $Z_0$ along this transition.

	\end{itemize}
\end{itemize}

\begin{lem}%
	\label{Lemma:PAComparator}
	Pushdown automaton $\A_{\PASuc}^\succeq$ accepts a pair of sequences $(A,B)$ iff $\PLA{A} \succeq \PLA{B}$.
\end{lem}
\begin{proof}[Proof sketch]
	To prove this statement, it is sufficient to demonstrate that $\A_{\PASuc}^\succeq$ accepts a pair of sequences $(A,B)$ iff there are only finitely many indexes where $\Sum{B[1,i]} > \Sum{A[1,i]}$. On $\A_{\PASuc}^\succeq$ this corresponds to the condition that there being only finitely many times when the PDA is in state $N$ during the run of $(A,B)$. This is ensured by the pushdown automaton since the word can be accepted only in state $F$ and there is no outgoing edge from $F$. Therefore, every word that is accepted by $\A_{\PASuc}^\succeq$ satisfies the condition $\ExistFin{i}{B}{A}$.

	Conversely, for every word $(A,B)$ that satisfies $\ExistFin{i}{B}{A}$ there is a point, call it index $k$, such that for all indexes $m>k$, $\Sum{B[1,m]} \ngeq \Sum{A[1,m]}$. If a run of $(A,B)$ switches to $F$ at this $m$, then it will be accepted by the  automaton. Since $\A_{\PASuc}^\succeq$ allows for non-deterministic move to $(F, \tau)$ from $(P, \tau)$, the run of $(A,B)$ will always be able to move to $F$ after index $m$. Hence, every $(A,B)$ satisfying $\ExistFin{i}{B}{A}$ will be accepted by $\A_{\PASuc}^\succeq$.
\end{proof}

\begin{thm}%
	\label{Thrm:PAComparator}
	The prefix-average comparator is an $\omega$-CFL\@.
\end{thm}
While $\omega$-CFL can be easily expressed, they do not possess closure properties, and several problems on $\omega$-CFL are easily undecidable. Hence, the application of $\omega$-context-free comparator will require further investigation.
For example, it is unclear whether $\omega$-context free comparators can solve quantitative inclusion since  complementation of $\omega$-CFL is undecidable. Problems like membership in $\omega$-CFG are decidable. We will have to investigate applications where reductions using $\omega$-context-free comparators require decidable operations such as  membership.


\section{Concluding remarks}%
\label{Sec:Conclusion}

In this paper, we identified a novel mode for comparison in quantitative systems: the online comparison of aggregate values of sequences of quantitative weights. This notion is embodied by comparators automata that read two infinite sequences of weights synchronously and relate their aggregate values.
We showed that all $\omega$-regular aggregate functions have $\omega$-regular comparators. However, the converse direction is still open: Are functions with $\omega$-regular comparators also $\omega$-regular?
We showed that $\omega$-regular comparators not only yield generic algorithms for problems including quantitative inclusion and winning strategies in incomplete-information quantitative games, they also result in algorithmic advances~\cite{TACAS21}. We establish that when the weights are represented in unary, the discounted-sum inclusion problem is $\cc{PSAPCE}$-complete for integer discount-factor, hence closing a complexity gap.
We showed that discounted-sum aggregate function and their comparators are $\omega$-regular iff the discount-factor $d>1$ is an integer.
We showed that prefix-average comparator are $\omega$-context-free.

We believe comparators, especially $\omega$-regular comparators, can be of significant utility in verification and synthesis of quantitative systems~\cite{bansal2020automata}, as demonstrated by the existence of finite-representation of counterexamples of the quantitative inclusion problem. Another potential application is computing equilibria in quantitative games. Applications of the prefix-average comparator, in general $\omega$-context-free comparators, is open to further investigation. Another direction to pursue is to study aggregate functions in more detail, and attempt to solve the conjecture relating $\omega$-regular aggregate functions and $\omega$-regular comparators.

\section*{Acknowledgment}
  \noindent We thank the anonymous reviewers for their comments. We  thank K. Chatterjee, L. Doyen, T. Henzinger, G. A. Perez and J. F. Raskin for corrections to earlier drafts, and their contributions to this paper.
  We thank P. Ganty and R. Majumdar for preliminary discussions on the limit-average comparator.
  This work was partially supported by NSF Grant No. 1704883, ``Formal Analysis and Synthesis of Multiagent Systems with Incentives''.


\bibliographystyle{alphaurl}
\bibliography{refs.bib,myRef.bib}

\end{document}